\documentclass{amsart}
\usepackage{enumerate}
\usepackage{amssymb}
\usepackage{latexsym,amscd,amsfonts}
\def \Z {\mathbb Z}
\def \R {\mathbb R}
\def \C {\mathbb C}

\def \N {\mathbb N}

\def \P {\mathbb P}
\def \E {\mathbb E}
\def\dd{\mathrm d}
\def\ee{\mathrm e}
\def\ii{\mathrm i}
\newcommand{\SpN}{\mathrm{Sp}_{\mathrm{N}}(\R)\,}
\newcommand{\spN}{\mathfrak{sp}_{\mathrm{N}}(\R)\,}

\newcommand{\supp}{\mathrm{supp}\,}
\newcommand{\Ho}{H(\omega)}
\newcommand{\Hl}{H_{\ell}(\omega)}
\newcommand{\HL}{H^{(L)}(\omega)}
\newcommand{\HLl}{H^{(L)}_{\ell}(\omega)}
\newcommand{\LpSI}{$p$-contracting and $L_p$-strongly irreducible, for every $p\in \{1,\ldots,N\}$}
\newcommand{\aN}{\{1,\ldots,N\}}
\newcommand{\dlO}{d_{\log\, \mathcal{O}}}
\newcommand{\omO}{\omega^{(0)}}

\newtheorem{thm}{Theorem}
\newtheorem{prop}{Proposition}
\newtheorem{cor}{Corollary}
\newtheorem{defi}{Definition}
\newtheorem{lem}{Lemma}
\begin{document}
       
\title{Localization for a matrix-valued Anderson model}
\author{Hakim Boumaza} 
\email{boumaza@math.keio.ac.jp}
\address{Keio University, Department of Mathematics\\
Hiyoshi 3-14-1, Kohoku-ku 223-8522\\
Yokohama, Japan\\}
\thanks{The author is supported by JSPS Grant P07728}
\begin{abstract}
We study localization properties for a class of one-dimensional, matrix-valued, continuous, random Schr\"odinger operators, acting on $L^2(\R)\otimes \C^N$, for  arbitrary $N\geq 1$. We prove that, under suitable assumptions on the F\"urstenberg group of these operators, valid on an interval $I\subset \R$, they exhibit localization properties on $I$, both in the spectral and dynamical sense. After looking at the regularity properties of the Lyapunov exponents and of the integrated density of states, we prove a Wegner estimate and apply a multiscale analysis scheme to prove localization for these operators. We also study an example in this class of operators, for which we can prove the required assumptions on the F\"urstenberg group. This group being the one generated by the transfer matrices, we can use, to prove these assumptions, an algebraic result on generating dense Lie subgroups in semisimple real connected Lie groups, due to Breuillard and Gelander. The algebraic methods used here allow us to handle with singular distributions of the random parameters.

\end{abstract}


\maketitle
\section{Introduction : models and results}\label{sec_intro}

In this paper, we will discuss localization properties of continuous matrix-valued Anderson models of the form : 
\begin{equation}\label{model_H}
\Ho = -\frac{\dd^2}{\dd x^2}\otimes I_{\mathrm{N}} + \sum_{n\in \Z} V_{\omega}^{(n)} (x-\ell n),
\end{equation}
\noindent acting on $L^2(\R)\otimes \C^N$, where $N\geq 1$ is an integer, $I_{\mathrm{N}}$ is the identity matrix  of order $N$ and $\ell>0$ is a real number. Let $(\Omega,\mathcal{A},\mathsf{P})$ be a complete probability space and let $\omega \in \Omega$. For every $n\in \Z$, the functions $x\mapsto V_{\omega}^{(n)} (x)$ are symmetric matrix-valued functions, supported on $[0,\ell]$ and bounded uniformely on $x$, $n$ and $\omega$. The sequence $(V_{\omega}^{(n)})_{n\in \Z}$ is a sequence of independent and identically distributed (\emph{i.i.d}) random variables on $\Omega$. We also assume that the potential $x\mapsto \sum_{n\in \Z} V_{\omega}^{(n)} (x-\ell n)$ is such that the operator $\Ho$ is $\Z$-ergodic. 

\noindent As a bounded perturbation of $-\frac{\dd^2}{\dd x^2}\otimes I_{\mathrm{N}}$, the operator $\Ho$ is self-adjoint on the Sobolev space  $H^2(\R)\otimes \C^N$ and thus, for every $\omega\in \Omega$, the spectrum of $\Ho$ is included in $\R$. 

\noindent Due to the hypothesis of $\Z$-ergodicity, there exists $\Sigma \subset \R$ such that, for $\mathsf{P}$-almost every $\omega \in \Omega$, $\Sigma=\sigma(\Ho)$. There also exist $\Sigma_{\mathrm{pp}}$, $\Sigma_{\mathrm{ac}}$ and $\Sigma_{\mathrm{sc}}$, subsets of $\R$, such that, for $\mathsf{P}$-almost every $\omega \in \Omega$, $\Sigma_{\mathrm{pp}}=\sigma_{\mathrm{pp}}(\Ho)$, $\Sigma_{\mathrm{ac}}=\sigma_{\mathrm{ac}}(\Ho)$ and $\Sigma_{\mathrm{sc}}=\sigma_{\mathrm{sc}}(\Ho)$.
\vskip 1mm

\noindent We will show that under suitable assumptions on the F\"urstenberg group of $\Ho$ (see Definition \ref{def_furstenberg}), this operator will exhibit localization properties on a certain interval of $\R$. These assumptions are not satisfied for every operators of the form (\ref{model_H}), but we will verify them for the following operator :
\begin{equation}\label{model_Hl}
\Hl = -\frac{\dd^2}{\dd x^2}\otimes I_{\mathrm{N}} + V_0 + \sum_{n\in \Z}  \left(
\begin{smallmatrix}
c_1 \omega_{1}^{(n)} \mathbf{1}_{[0,\ell]}(x-\ell n) & & 0\\ 
 & \ddots &  \\
0 & & c_N \omega_{N}^{(n)} \mathbf{1}_{[0,\ell]}(x-\ell n)\\ 
\end{smallmatrix}\right),
\end{equation}
\noindent acting on $L^2(\R)\otimes \C^N$. The real number $\ell >0$ represents the length of the range of the random interactions. The constants $c_1,\ldots, c_N$ are non-zero real numbers and $V_0$ is the multiplication operator by the tridiagonal matrix $V_0$ having a null diagonal and coefficients on the upper and lower diagonals all equal to $1$.

\noindent For every $i\in \aN$, the $(\omega_i^{(n)})_{n\in \Z}$ are sequences of \emph{i.i.d.} random variables on a complete probability space $(\widetilde{\Omega},\widetilde{\mathcal{A}},\widetilde{\mathsf{P}})$, of common law $\nu$ such that $\{0,1\} \subset \supp \nu$ and $\supp \nu$ is bounded. The random parameter $\omega$ is an element of the product space 
$$(\Omega,\mathcal{A},\mathsf{P})=\left(\otimes_{n\in \Z}\widetilde{\Omega}^{\otimes N},\otimes_{n\in \Z}\widetilde{\mathcal{A}}^{\otimes N}, \otimes_{n\in \Z}\widetilde{\mathsf{P}}^{\otimes N}\right)$$
and we also set, for every $n\in \Z$, $\omega^{(n)}=(\omega_1^{(n)},\ldots,\omega_N^{(n)})$,
of law $\nu^{\otimes N}$. The expectancy against $\mathsf{P}$ will be denoted by $\E(.)$.
\vskip 1mm

\noindent  The model (\ref{model_Hl}) is a particular case of (\ref{model_H}). Indeed, $\Hl$ is $\Z$-ergodic and the potential part of $\Hl$ is uniformly bounded on $x$, $n$ and $\omega$ because of the boundedness of $\supp \nu$.
\vskip 5mm

\noindent Following \cite{K07}, we give the definitions of localization properties for $\Ho$, from the spectral and the dynamical point of views. For $x\in \R$, we denote by $\mathbf{1}_{x}$ the characteristic function of the interval of length $2\ell$ centered at $x$. We also write $<x>=\sqrt{1+|x|^2}$ and we denote by $E_{\omega}(.)$ the spectral projection of $\Ho$. The Hilbert-Schmidt norm is written as $||\ ||_2$ while the $L^2$-norm is written as $||\ ||$.

\begin{defi}\label{def_localization}
Let $I\subset \R$ be an open interval. We say that :
\begin{itemize}
\item[(i)] $\Ho$ exhibits \emph{exponential localization} (EL) in $I$, if it has pure point spectrum in $I$ (\emph{i.e.}, $\Sigma\cap I=\Sigma_{\mathrm{pp}}\cap I$ and $\Sigma_{\mathrm{ac}}\cap I=\Sigma_{\mathrm{sc}}\cap I=\emptyset$) and, for $\mathsf{P}$-almost every $\omega$, the eigenfunctions of $\Ho$ with eigenvalues in $I$ decay exponentially in the $L^2$-sense (\emph{i.e.}, there exist $C$ and $m>0$ such that $||\mathbf{1}_{x} \psi || \leq C\ee^{-m|x|}$ for $\psi$ eigenfunction of $\Ho$) ;
\item[(ii)] $\Ho$ exhibits \emph{strong dynamical localization} (SDL) in $I$, if $\Sigma \cap I \neq \emptyset$ and, for each compact interval $\tilde{I} \subset I$ and $\psi \in L^2(\R)\otimes C^N$ with compact support, we have,
$$\forall n\geq 0,\ \E\left( \sup_{t\in \R} \left|\left| <x>^{\frac{n}{2}} E_{\omega} (\tilde{I}) \ee^{-\ii t \Ho} \psi\right|\right|^2 \right) <\infty \ ;$$
\item[(iii)] $\Ho$ exhibits \emph{strong sub-exponential HS-kernel decay} (SSEHSKD) in $I$ if $\Sigma \cap I \neq \emptyset$ and, for each compact interval $\tilde{I} \subset I$ and $0<\zeta <1$, there is a finite constant $C_{\tilde{I},\zeta}$ such that,
$$\forall x,y\in \Z,\ \E\left( \sup_{||f||_{\infty} \leq 1}\left|\left| \mathbf{1}_{x} E_{\omega}(\tilde{I}) f(\Ho) \mathbf{1}_{y} \right|\right|_2^2 \right) \leq C_{\tilde{I},\zeta} \ee^{-|x-y|^{\zeta}},$$
$f$ being a bounded Borel function on $\R$ and $||f||_{\infty}=\sup_{t\in \R} |f(t)|$.
\end{itemize}
\end{defi}

\noindent We also set $\Sigma_{\mathrm{EL}}$, $\Sigma_{\mathrm{SDL}}$ and $\Sigma_{\mathrm{SSEHSKD}}$ as the sets of $E\in \Sigma$ for which there exists an open interval $I$, $E\in I$, such that $\Ho$ exhibits on $I$, (EL), (SDL) and (SSEHSKD) respectively. We have $\Sigma_{\mathrm{SSEHSKD}} \subset \Sigma_{\mathrm{SDL}}$ and we will actually prove (SSEHSKD) for $\Hl$ on some interval, which will imply (SDL) on the same interval. We quote (SDL) property as it has a more natural interpretation than (SSEHSKD) in terms of control of the moments of the wave packets of $\Ho$.
\vskip 5mm

\noindent We are now ready to give the statement of our main results. For $E\in \R$, let $G(E)$ be the F\"urstenberg group associated to $\Ho$ (see Definition \ref{def_furstenberg}). For the definitions of $p$-contractivity and $L_p$-strong irreducibility, see Definition \ref{def_pcon_LpSI}.

\begin{thm}\label{thm_localization}
Let $I\subset \R$ be a compact interval such that $\Sigma\cap I\neq \emptyset$ and let $\tilde{I}$ be an open interval, $I\subset \tilde{I}$, such that, for every $E\in \tilde{I}$, $G(E)$ is \LpSI. Then, $\Ho$ exhibits (EL), (SDL) and (SSEHSKD) in $I$.
\end{thm}
\vskip 2mm

\noindent Before applying this theorem to the operator $\Hl$, we need to introduce some notations. Let $\SpN$ denote the group of $2N\times 2N$ real symplectic matrices and let $\mathcal{O}$ be the neighborhood of $I_{2\mathrm{N}}$ in $\SpN$ given by Theorem \ref{thm_breuillard} applied to $G=\SpN$. We set : 
$$\dlO=\max\{ R>0\ |\ B(0,R)\subset \log\,\mathcal{O} \},$$
where $B(0,R)$ is the open ball, centered on $0$ and of radius $R>0$, for the topology induced on the Lie algebra $\spN$ of $\SpN$ by the matrix norm induced by the euclidean norm on $R^{2N}$. 

\noindent For $\omO=(\omega_1^{(0)},\ldots,\omega_N^{(0)})\in \{0,1\}^N$, let 
$$M_{\omO} = V_0 + \mathrm{diag}(c_1 \omega_1^{(0)},\ldots,c_N \omega_N^{(0)}).$$
As $M_{\omO}$ is a real symmetric matrix, it has $\lambda_1^{\omO},\ldots ,\lambda_N^{\omO}$ as real eigenvalues. We set,
\begin{equation}\label{eq_lambda_minmax}
\lambda_{\mathrm{min}}=\min_{\omO\in \{ 0,1\}^N} \min_{1\leq i\leq N} \lambda_i^{\omO},\qquad \lambda_{\mathrm{max}}=\max_{\omO\in \{ 0,1\}^N} \max_{1\leq i\leq N} \lambda_i^{\omO}
\end{equation}
and $\delta_0=\frac{\lambda_{\mathrm{max}}-\lambda_{\mathrm{min}}}{2}$. We also set 
\begin{equation}\label{eq_def_lC}
\ell_C:=\ell_C(N)=\min \left( 1, \frac{\dlO}{\delta_0}\right)
\end{equation}
and, for every $\ell<\ell_C$,
\begin{equation}\label{eq_def_IlN}
I(\ell,N)=\left[ \lambda_{\mathrm{max}} - \frac{\dlO}{\ell},\lambda_{\mathrm{min}} + \frac{\dlO}{\ell} \right].
\end{equation}
\vskip 3mm

\noindent Applying Theorem \ref{thm_localization} to the operator $\Hl$, we obtain the following results.
\vskip 1mm

\begin{thm}\label{thm_localization_Hl}
\ \begin{itemize}
\item[(i)] Assume that $\ell <\ell_C$ and let $I\subset I(\ell,N)$ be an open interval such that $\Sigma \cap I \neq \emptyset$. Then $\Hl$ exhibits (EL), (SDL) and (SSEHSKD) in $I$.
\item[(ii)] Assume that $\ell=1$ and $N=2$ in (\ref{model_Hl}). There exists a discrete set $\mathcal{S} \subset \R$ such that, for every compact interval $I\subset (2,+\infty)\setminus \mathcal{S}$ with $\Sigma \cap I \neq \emptyset$, $H_{1}(\omega)$ exhibits (EL), (SDL) and (SSEHSKD) in $I$. 
\end{itemize}
\end{thm}
\vskip 2mm

\noindent We remark that in point $(i)$ of Theorem \ref{thm_localization_Hl}, as the length of $I(\ell,N)$ tends to $+\infty$ when $\ell$ tends to $0^+$, taking $\ell$ small enough ensure that $\Sigma \cap I(\ell,N) \neq \emptyset$ and, moreover, we can always find a non-trivial open interval $I\subset I(\ell,N)$ such that $\Sigma \cap I \neq \emptyset$.
\vskip 2mm

\noindent To prove localization results as Theorem \ref{thm_localization} and Theorem \ref{thm_localization_Hl} for one-dimensional operators such as (\ref{model_H}) and (\ref{model_Hl}), we can follow this plan : 

\begin{enumerate}[1.]
\item We prove that the Lyapunov exponents of $\Ho$ are all distinct and positive.
\item We prove the H\"older regularity of these exponents.
\item We deduce the same H\"older regularity for the integrated density of states of $\Ho$.
\item With this regularity of the integrated density of states, we prove a Wegner estimate.
\item We apply a multiscale analysis scheme.
\end{enumerate}
\vskip 2mm

\noindent According to this plan, our first result for $\Hl$ is the following. For the definitions of $\mu_E$ and $L_p$, see Section \ref{sec_lyap_prop}.

\begin{thm}\label{thm_lyap_N}
Assume that $\ell <\ell_C$. Then,
\begin{itemize}
\item[(i)] the $N$ positive Lyapunov exponents of $\Hl$, $\gamma_1(E),\ldots,\gamma_N(E)$, verify
\begin{equation}\label{eq_lyap_pos}
\forall E\in I(\ell,N),\ \gamma_1(E)> \cdots > \gamma_N(E)>0.
\end{equation}
Therefore, $\Hl$ has no absolutely continuous spectrum in $I(\ell,N)$, \emph{i.e.}, $\Sigma_{\mathrm{ac}} \cap I(\ell,N) =\emptyset$.
\item[(ii)] For every $p\in \aN$, there exists a unique $\mu_E$-invariant measure $\nu_{p,E}$ on $\P(L_p)=\{\bar{x} \in \P(\wedge^p \R^{2N})\ |\ x\in L_p \}$ such that, for every $E\in I(\ell,N)$,
\begin{equation}\label{eq_lyap_rep_int}
\sum_{i=1}^{p} \gamma_i(E)= \int_{\SpN \times \P(L_p)} \log \frac{||(\wedge^p M)x||}{||x||} \dd \mu_E(M) \dd \nu_{p,E}(\bar{x}).
\end{equation}
\item[(iii)] For every $i\in \aN$, $E\mapsto \gamma_i(E)$ is H\"older continuous on $I(\ell,N)$, \emph{i.e.}, there exist $C>0$ and $\alpha >0$ such that,
\begin{equation}\label{eq_lyap_reg}
\forall E,E'\in I(\ell,N),\ |\gamma_i(E)-\gamma_i(E')|\leq C |E-E'|^{\alpha}.
\end{equation}
\end{itemize}
\end{thm}
\vskip 3mm

\noindent Points $(i)$ and $(ii)$ will directly follow from the theory of sequences of \emph{i.i.d.} random matrices in $\SpN$ after proving that, for every $E\in I(\ell,N)$, the F\"urstenberg group of $\Hl$ is \LpSI. It is exactly the assumption of Theorem \ref{thm_localization}. Then, applying results of \cite{boumazarmp}, we obtain a regularity result for the integrated density of states $E\mapsto N(E)$.

\begin{thm}\label{thm_ids}
Let $\ell <\ell_C$. Let $I$ be a compact interval included in the interior of $I(\ell,N)$. The integrated density of states of $\Hl$ is H\"older continuous on $I$, \emph{i.e.}, there exist $C>0$ and $\alpha >0$ such that,
\begin{equation}\label{eq_ids_reg}
\forall E,E'\in I,\ |N(E)-N(E')|\leq C |E-E'|^{\alpha}.
\end{equation}
\end{thm}
\vskip 3mm

\noindent The local H\"older regularity of the integrated density of states is a key ingredient to prove a Wegner estimate for $\Ho$. Let $L\in \N^*$ and denote by $\HL$ the restriction of $\Ho$ to $L^2([-\ell L,\ell L])\otimes \C^N$ with Dirichlet boundary conditions. We define $\HLl$ the same way, for every $\ell >0$.

\begin{thm}\label{thm_wegner}
Let $I\subset \R$ be a compact interval and $\tilde{I}$ be an open interval, $I \subset \tilde{I}$, such that, for every $E\in \tilde{I}$, $G(E)$ is \LpSI. Then, for every $\beta \in (0,1)$ and every $\kappa >0$, there exist $L_0\in \N$ and $\xi>0$ such that,
\begin{equation}\label{eq_thm_wegner}
\mathsf{P}\left( d\left(E, \sigma(\HL) \right) \leq \ee^{-\kappa(\ell L)^{\beta}} \right) \leq \ee^{-\xi(\ell L)^{\beta}},
\end{equation}
for every $E\in I$ and every $L\geq L_0$.
\end{thm}
\vskip 2mm

\noindent Applying this general result to $\Hl$ and using Theorem \ref{thm_lyap_N}, lead us to the following.

\begin{cor}\label{cor_wegner}
\
\begin{itemize}
\item[(i)] Let $\ell <\ell_C$. Then (\ref{eq_thm_wegner}) holds for $\HLl$ for every $E\in I$, where $I$ is a non-trivial interval included in the interior of $I(\ell,N)$.
\item[(ii)] Let $\ell=1$ and $N=2$ in (\ref{model_Hl}). There exists a discrete set $\mathcal{S} \subset \R$ such that, for every compact interval $I\subset (2,+\infty) \setminus \mathcal{S}$, (\ref{eq_thm_wegner}) holds for $H_{1}^{(L)}(\omega)$, for every $E\in I$.
\end{itemize}
\end{cor}

\noindent Then, to obtain Theorem \ref{thm_localization} and Theorem \ref{thm_localization_Hl}, it will remain to show that we can apply a multiscale analysis scheme as presented in \cite{stollmann} or \cite{K07}.
\vskip 5mm

\noindent After \cite{DSS02}, we already know that, in the scalar-valued case (corresponding here to $N=1$), there exists a discrete set $\mathcal{S}\subset \R$ such that, on every compact interval $I\subset \R \setminus \mathcal{S}$, $\Sigma \cap I \neq \emptyset$, we have exponential localization and strong dynamical localization. Thus, there is localization at small and large energies in this case. On higher dimension $d\geq 1$, the analog of model (\ref{model_H}) is the operator
\begin{equation}\label{model_dimd}
\Ho = -\Delta_d\otimes I_{\mathrm{N}} + \sum_{n\in \Z} V_{\omega}^{(n)} (x-\ell n),
\end{equation}
acting on $L^2(\R^d)\otimes \C^N$, with the same assumptions as for model (\ref{model_H}). In particular, we are interested in handling singular distributions of the random parameter $\omega$. For $N=1$ and $d\geq 2$, it is known since \cite{BK05} that there is localization, in the spectral and the dynamical sense, near the bottom of the almost-sure spectrum of $\Ho$. For large energies, the question of the localization, in the case of singular distribution of $\omega$, is still open. It is commonly conjectured that for $d\geq 3$, there exist delocalized states at large energies. On the contrary, for $d=2$, it is conjectured that there is localization, even for large energies, exactly like in the case $d=1$. To tackle the question of localization for $d=2$, including singular distribution of $\omega$, we can start by studying a simpler model, a continuous strip in $\R^2$. We consider the restriction of (\ref{model_dimd}) to the continuous strip $\R \times [0,1]$, acting on $L^2(\R \times [0,1])\otimes \C^N$, with Dirichlet boundary conditions on $\R \times \{0\}$ and $\R\times \{1\}$. This continuous strip model is not only interesting as a first step to study localization on $\R^2$, but it is also of physical interest. Indeed, such a model can be considered to modelize nanoconductors and it allows to study their transport properties. 

\noindent The question of localization at all energies for the continuous strip $\R\times [0,1]$ is a difficult problem, mostly due to its PDE's nature. A possible approach for this question is to operate a discretization in the bounded direction. It leads to consider a quasi-one dimensional model which is matrix-valued, and actually of the form (\ref{model_H}). Having a one-dimensional model turns the nature of the problem to an ODE's one, which allows to use formalism such as transfer matrices and Lyapunov exponents. Then, we want to obtain localization for arbitrary large $N$ and we hope to be able to recover localization properties for the continuous strip in dimension $2$, by letting $N$ ``tends to infinity''. For this purpose, it is important to have results for arbitrary $N\geq 1$.

\noindent In a previous article of the author, \cite{boumazampag}, we proved separability of the Lyapunov exponents of $H_1(\omega)$ for large energies, but only for $N=2$. It was done by proving $p$-contractivity and $L_p$-strong irreducibility of the F\"urstenberg group for energies $E>2$ (for $N=2$, $\lambda_{\mathrm{max}}=2$) and away from a discrete set of $\R$. Point $(ii)$ of Theorem \ref{thm_localization_Hl} is based upon this result. Due to some technical difficulties, it was not possible to generalize the computations done for $N=2$ to an arbitrary $N\geq 1$. This is were the parameter $\ell >0$ play an important role, as explained in the end of Section \ref{sec_proof_thm_lyap_N}. The main difference between what was proved for $\ell=1$ and $N=2$ in \cite{boumazampag}, and the case $N\geq 1$ for $\ell$ small considered here, is the existence of a discrete set of critical energies which does not appear in the second case.

\noindent We also want to mention that different methods have been used in \cite{KMPV03} to prove localization for random operators on strips. These methods, using spectral averaging techniques, do not apply for singular distribution of the random parameters. We choose instead to follow methods of \cite{KLS90} for the discrete strip and adapt them to our models. The same strategy was already followed in \cite{DSS02} for the scalar-valued case.
\vskip 3mm

\noindent We finish this introduction by giving the outline of the rest of the article. In Section \ref{sec_lyap}, we prove Theorem \ref{thm_lyap_N}. We first recall definitions of the Lyapunov exponents in Section \ref{sec_lyap_prop} before introducing the F\"urstenberg group of $\Hl$ and study it in Section \ref{sec_proof_thm_lyap_N}. In Section \ref{sec_estimate_trans_mat}, we review and adapt estimates on the random walk defined by the transfer matrices and, in particular, large deviation type estimates. Then we shortly discuss the H\"older regularity of the integrated of states in Section \ref{sec_ids}. We can deduce from this regularity result a Wegner estimate, as it is done in Section \ref{sec_wegner}. Finally, in Section \ref{sec_localization}, we give the proofs of Theorem \ref{thm_localization} and of Theorem \ref{thm_localization_Hl}. We start by presenting the requirements of a multiscale analysis scheme in Section \ref{sec_MSA} and then we prove, in Section \ref{sec_proof_localization}, an initial length scale estimate required to apply this scheme.

\section{Positivity and regularity of the Lyapunov exponents}\label{sec_lyap}

In this Section, we will give the proof of Theorem \ref{thm_lyap_N}. Before that, we will set up notations and definitions about Lyapunov exponents.

\subsection{Lyapunov exponents}\label{sec_lyap_prop}

We start with a review of the definition of the Lyapunov exponents. Let $N$ be a positive integer and let $\SpN$ denote the group of $2N\times 2N$ real symplectic matrices. It is the subgroup of $\mathrm{GL}_{2\mathrm{N}}(\R)$ of matrices $M$ satisfying $$^tMJM=J,$$ 
where $J$ is the matrix of order $2N$ defined by $J=\bigl(\begin{smallmatrix}
0 & -I_{N} \\
I_{N} & 0
\end{smallmatrix}\bigr)$.

\begin{defi}
Let $(T_{n}^{\omega})_{n\in \N}$ be a sequence of \emph{i.i.d.}\ random matrices in $\SpN$ with $\E(\log^{+}||T_{0}^{\omega}||) <\infty$. The Lyapounov exponents $\gamma_{1},\ldots,\gamma_{2N}$ associated with $(T_{n}^{\omega})_{n\in \N}$ are defined inductively, for every $p\in \aN$, by \begin{equation}\label{eq_lyap_def}
\sum_{i=1}^{p} \gamma_{i} = \lim_{n \to \infty} \frac{1}{n}
\mathbb{E}(\log ||\wedge^{p} (T_{n-1}^{\omega}\ldots T_{0}^{\omega})||).
\end{equation}
\end{defi}

\noindent Here, $\wedge^{p} M$ denotes the $p$th exterior power of the matrix $M$, acting on the $p$th exterior power of $\R^{2N}$. One has $\gamma_{1}\geq \ldots \geq \gamma_{2N}$. Moreover, we have the symmetry property $\gamma_{2N-i+1}= -\gamma_{i}$, for every $i \in \{1,\ldots,N\}$, due to the symplecticity of the random matrices $T_{n}^{\omega}$. We will only have to study the $N$ first Lyapunov exponents, those being positive.

\noindent We also define, for every $p\in \aN$, the $p$-Lagrangian submanifold $L_{p}$ of $\R^{2N}$, as the subspace of $\wedge^{p}\R^{2N}$ spanned by $\{Me_1 \wedge \ldots \wedge Me_p \ |\ M\in \SpN \}$, where $(e_{1},\ldots,e_{2N})$ is the canonical basis of $\R^{2N}$. We note that $L_1=\wedge^1 \R^{2N}=\R^{2N}$. We can now give the definitions of $p$-contractivity and $L_p$- strong irreducibility.

\begin{defi}\label{def_pcon_LpSI}
Let $G$ be a subset of $\SpN$ and $p\in \aN$.
\begin{itemize}
\item[(i)] $G$ is \emph{$p$-contracting} if there exists a sequence $(T_n)_{n\in \N}$ in $G$ such that the sequence $(||\wedge^p T_n||^{-1}\wedge^p T_n)_{n\in \N}$ converges to a rank-one matrix.
\item[(ii)] $G$ is \emph{$L_p$-strongly irreducible} if there does not exist a finite union $W$ of proper subspaces of $L_p$ such that, $(\wedge^p T)(W)=W$ for any $T\in G$.
\end{itemize}
\end{defi} 

\noindent We can now give the proof of Theorem \ref{thm_lyap_N}.

\subsection{Proof of Theorem \ref{thm_lyap_N}}\label{sec_proof_thm_lyap_N}

In this Section, all the definitions are given for the operator $\Hl$, but we could also define the same objects for the more general operator (\ref{model_H}). We start by introducing the sequence of transfer matrices associated to the operator $\Hl$. Let $E\in \R$. The transfer matrix $T_{\omega^{(n)}}(E)$ of $\Hl$, from $\ell n$ to $\ell (n+1)$, is defined by the relation
\begin{equation}\label{eq_def_trans_mat}
\forall n\in \Z,\ \left( \begin{array}{c}
u(\ell(n+1)) \\
u'(\ell(n+1)) 
\end{array} \right) = T_{\omega^{(n)}}(E) \left( \begin{array}{c}
u(\ell n) \\
u'(\ell n) 
\end{array} \right),
\end{equation}
\noindent where $u:\R\to \C^N$ is a solution of the second order differential system
\begin{equation}\label{eq_sys_Hl}
\Hl u=Eu.
\end{equation}
\noindent We can give the explicit form of the matrices $T_{\omega^{(n)}}(E)$. For $E\in \R$, $n\in \Z$ and $\omega^{(n)}\in \tilde{\Omega}^{\otimes N}$, we set
\begin{equation}\label{eq_def_mat_M}
M_{\omega^{(n)}}(E)=V_0 + \mathrm{diag}(c_1 \omega_1^{(n)},\ldots,c_N \omega_N^{(n)})-EI_{\mathrm{N}}. 
\end{equation}
\noindent Then, if we set 
\begin{equation}\label{eq_def_mat_X}
X_{\omega^{(n)}}(E)=\left( \begin{array}{cc}
0 & I_{\mathrm{N}} \\
M_{\omega^{(n)}}(E) & 0
\end{array}\right) \in \mathcal{M}_{\mathrm{2N}}(\R),
\end{equation}
by solving the constant coefficient system (\ref{eq_sys_Hl}) on $[\ell n,\ell (n+1)]$, we have :
\begin{equation}\label{eq_def_mat_T}
\forall \ell >0,\ \forall n\in \Z,\ \forall E\in \R,\ T_{\omega^{(n)}}(E)=\exp\left(\ell X_{\omega^{(n)}}(E)\right).
\end{equation}
\vskip 2mm

\noindent The fact that $T_{\omega^{(n)}}(E)$ is the exponential of a matrix will be very important to be able to apply Theorem \ref{thm_breuillard}. We can now introduce the central object involved in the proof of Theorem \ref{thm_lyap_N}. It is the algebraic object containing all the products of transfer matrices.

\begin{defi}\label{def_furstenberg}
For every real number $E\in \R$, the \emph{F\"urstenberg group} of $\Hl$ is defined by 
$$G(E)=\overline{<\supp \mu_E>},$$
where $\mu_E$ is the common distribution of the $T_{\omega^{(n)}}(E)$.
\end{defi}

\noindent As the $T_{\omega^{(n)}}(E)$ are \emph{i.i.d.}, $\mu_E=(T_{\omO}(E))_{*}\, \nu^{\otimes N}$ and we have the internal description of $G(E)$ :
\begin{equation}\label{eq_description_GE}
\forall E\in \R,\ G(E)=\overline{<T_{\omO} (E)\ |\ \omO \in \supp \nu^{\otimes N}>}.
\end{equation}
\vskip 2mm

\noindent As $\{0,1\}\subset \supp \nu$, we also have $G(E) \supset \overline{<T_{\omO} (E)\ |\ \omO \in \{0,1\}^N >}$. Due to a criterion of Gol'dsheid and Margulis (see \cite{GM89,boumazastolz}), to prove that, for a given $E\in \R$, $G(E)$ is \LpSI, it suffices to prove that $G(E)$ is Zariski dense in $\SpN$. Actually, we will prove a stronger statement which is that, for every $E\in I(\ell,N)$, $G(E)$ is equal to $\SpN$. Therefore, for every $E\in I(\ell,N)$, $G(E)$ will be \LpSI.

\begin{prop}\label{prop_GE_SpN}
Let $\ell <\ell_C$ and $I(\ell,N)$ be the interval defined at (\ref{eq_def_IlN}). Then, for every $E\in I(\ell,N)$, $G(E)=\SpN$.
\end{prop}

\noindent The proof of this proposition is based upon the following theorem due to Breuillard and Gelander.

\begin{thm}[\cite{BG03}, Theorem 2.1]\label{thm_breuillard}
Let $G$ be a real, connected, semisimple Lie group, whose Lie algebra is $\mathfrak{g}$. Then, there is a neighborhood $\mathcal{O}$ of $1$ in $G$, on which $\log=\exp^{-1}$ is a well defined diffeomorphism, such that $g_{1},\ldots,g_{m}\in \mathcal{O}$ generate a dense subgroup whenever $\log g_{1},\ldots,\log g_{m}$ generate $\mathfrak{g}$.
\end{thm}

\noindent This theorem gives us the outline of the proof of Proposition \ref{prop_GE_SpN} :
\begin{enumerate}[1.]
\item We prove that, for every $\ell \in (0,\ell_C)$ and every $E\in I(\ell,N)$, $T_{\omO}(E)\in \mathcal{O}$, for every $\omO \in \{0,1\}^N$.
\item For $\ell <\ell_C$, we compute $\log T_{\omO}(E)$.
\item We prove that $\mathrm{Lie}\{\log T_{\omO}(E)\ |\ \omO \in \{0,1\}^N \}=\spN$, the Lie algebra of $\SpN$.
\end{enumerate}
\vskip 2mm

\noindent Before proving Proposition \ref{prop_GE_SpN}, we prove the following algebraic lemma which will be used to prove point $3$.

\begin{lem}\label{lem_a_spN}
Let $N\geq 1$ and $E\in \R$. The Lie algebra generated by $\{ X_{\omO}(E) \ |\ \omO \in \{ 0,1\}^N \}$ is equal to $\spN$.
\end{lem}

\begin{proof}
First, we recall that :
$$\spN = \left\{ \left( \begin{array}{cc}
a & b_{1} \\
b_{2} & -^{t}a
\end{array} \right),\ a\in \mathcal{M}_{\mathrm{N}}(\R),\ b_{1}\ \mathrm{and}\ b_{2}\ \mathrm{symmetric} \right\}.$$
For $i,j\in \aN$, let $E_{ij}$ be the matrix in $\mathcal{M}_{\mathrm{N}}(\R)$ with a $1$ coefficient at the intersection of the $i$th row and the $j$th column, and $0$ elsewhere. We also set
$$\forall i,j\in \aN,\ X_{ij}=\frac{1}{2} \left( \begin{array}{cc}
0 & E_{ij}+E_{ji} \\
0 & 0
\end{array} \right),\ Y_{ij}={^t}X_{ij},\ Z_{ij}= \left( \begin{array}{cc}
E_{ij} & 0  \\
0 & -E_{ji}
\end{array} \right).$$
We also denote by $\delta_{ij}$ the Kronecker's symbol : 
$$\delta_{ij}=\left\lbrace \begin{array}{ccl}
1 & \mathrm{if} & i=j \\
0 & \mathrm{if} & i\neq j.
\end{array} \right.$$
We remark that the set $\{ X_{ij},Y_{ij},Z_{ij} \}_{i,j=1..N}$ is a basis of $\spN$. By direct computation, we get the relations, for every $i,j,k,r \in \aN$, 
\begin{itemize}
\item[(i)] $[Z_{ij},X_{kr}]=\delta_{jk} X_{ir}+\delta_{jr} X_{ik}$ 
\item[(ii)] $[Y_{kr},Z_{ij}]= \delta_{ik} Y_{rj}+\delta_{ir} Y_{kj}$ 
\item[(iii)] $[X_{ij},Y_{kr}]= \frac{1}{4} (\delta_{jk}Z_{ir}+\delta_{jr}Z_{ik}+\delta_{ki}Z_{jr}+ \delta_{ir}Z_{jk})$
\end{itemize}
\vskip 1mm

\noindent where $[\ ,\ ]$ is the usual bracket on Lie algebra of linear Lie groups. From these relations, we deduce that $\spN$ is generated by 
$$\big\{ X_{ij},Y_{ij}\ |\ i,j\in \aN, |i-j|\leq 1\big\}.$$
Indeed, let $\mathfrak{g}$ be the Lie algebra generated by this set. Let $i\in\aN$. Then, $Z_{ii}=2[X_{ii},Y_{ii}]\in \mathfrak{g}$ and $Z_{i,i+1}=2[X_{ii},Y_{i,i+1}] \in \mathfrak{g}$. Thus, for every $i,j\in \aN$, $|i-j|\leq 1$, $Z_{ij}\in \mathfrak{g}$. Then we have, $X_{i,i+2}=[Z_{i,i+1},Y_{i+1,i+2}]$, $Y_{i,i+2}= [Y_{i,i+1},Z_{i+1,i+2}]\in \mathfrak{g}$ and $Z_{i,i+2}=2[X_{i,i+1},Y_{i+1,i+2}] \in \mathfrak{g}$. Thus, for every $i,j\in \aN$, $|i-j|=2$, $X_{ij}$, $Y_{ij}$, $Z_{ij}\in \mathfrak{g}$. By induction, we do the same for indices $i,j$ such that $|i-j|=3$ and more generally for all indices $i,j\in \aN$. Thus, we proved that $\{ X_{ij},Y_{ij},Z_{ij} \}_{i,j=1..N}$ is included in $\mathfrak{g}$ and then $\spN\subset \mathfrak{g}$. Finally, $\mathfrak{g}=\spN$.
\vskip 2mm

\noindent According to this, to prove Lemma \ref{lem_a_spN}, we only have to prove that, for every $E\in \R$, the Lie algebra generated by $\{ X_{\omO}(E)\ |\ \omO \in \{0,1\}^N \}$ contains all the matrices $X_{ij}$ and $Y_{ij}$ for $i,j\in \aN$, $|i-j|\leq 1$. Let 
\begin{equation}\label{eq_def_lie_trans_mat}
\mathfrak{a}(E)=\mathrm{Lie} \big\{ X_{\omO}(E)\ \big|\ \omO \in \{0,1\}^N \big\}.
\end{equation}
To prove that $\mathfrak{a}(E)$ contains the matrices $X_{ij}$ and $Y_{ij}$ for $i,j\in \aN$, $|i-j|\leq 1$, we will proceed in several steps. We fix $E\in \R$.
\vskip 2mm

\noindent \textbf{Step 1.} We prove that the matrices $Z_{ii}$ for $i\in \aN$ are in $\mathfrak{a}(E)$. Let $\omO$ and $\tilde{\omega}^{(0)}$ in $\{0,1\}^N$. We have :
$$[X_{\omO}(E),X_{\tilde{\omega}^{(0)}}(E)]  =  X_{\omO}(E)X_{\tilde{\omega}^{(0)}}(E)- X_{\tilde{\omega}^{(0)}}(E)X_{\omO}(E)\qquad \qquad \qquad \qquad \qquad $$
$$ = \mathrm{diag}(c_1(\tilde{\omega}_1^{(0)}-\omega_1^{(0)}),\ldots,c_N(\tilde{\omega}_N^{(0)}- \omega_N^{(0)}),c_1(\omega_1^{(0)}-\tilde{\omega}_1^{(0)}),\ldots, c_N(\omega_N^{(0)}-\tilde{\omega}_N^{(0)})).$$
\vskip 1mm

\noindent In particular, for $\omO=(0,\ldots,0)$ and $\tilde{\omega}^{(0)}=(0,\ldots,1,\ldots,0)$, with a $1$ at the $i$th place and $0$ elsewhere, we get $Z_{ii}=[X_{\omO}(E),X_{\tilde{\omega}^{(0)}}(E)] \in \mathfrak{a}(E)$.
\vskip 1mm

\noindent \textbf{Step 2.} With the same choice of $\omO$ and $\tilde{\omega}^{(0)}$, we get
$$X_{\tilde{\omega}^{(0)}}(E)-X_{\omO}(E)=Y_{ii}.$$
Thus, for every $i\in \aN$, $Y_{ii}\in \mathfrak{a}(E)$.
\vskip 1mm

\noindent \textbf{Step 3.} We fix $\omO\in \{ 0,1\}^N$ and $i\in \aN$. We have :
\begin{eqnarray}
[X_{\omO}(E),Z_{ii}] & = & \left( \begin{array}{cc}
 0 & -2E_{ii} \\
 M_{\omO}(0)E_{ii}+E_{ii}M_{\omO}(0)-2EE_{ii} & 0
 \end{array} \right) \nonumber \\
  & = & -2X_{ii} + 2Y_{i,i-1} + 2Y_{i,i+1} + 2(\omega_{i}^{(0)}-E) Y_{ii} \nonumber
\end{eqnarray}
with the convention that $Y_{ij}$ is zero if the index $j$ is not in $\{1,\ldots ,N\}$. Thus, dividing by $2$, one gets,
\begin{equation}\label{eq_lem_a_1}
\forall i\in \aN,\ -X_{ii} + Y_{i,i-1} + Y_{i,i+1} + (\omega_{i}^{(0)}-E) Y_{ii} \in \mathfrak{a}(E).
\end{equation}
\vskip 1mm

\noindent \textbf{Step 4.} We prove that the matrix $J$ is in $\mathfrak{a}(E)$. We fix $\omO=(0,\ldots,0)$. By summing (\ref{eq_lem_a_1}) for $i\in \aN$, we stay in $\mathfrak{a}(E)$ and we have :
$$\sum_{i=1}^{N} \left(-X_{ii}+ Y_{i,i-1} + Y_{i,i+1}-EY_{ii}\right) = \sum_{i=1}^{N} (-X_{ii}) + \left( \begin{array}{cc}
0 & 0 \\
M_{\omO}(E) & 0
\end{array} \right) \in \mathfrak{a}(E).$$
We can substract $X_{\omO}(E)\in \mathfrak{a}(E)$ from this, to get : 
$$\sum_{i=1}^{N} (-X_{ii}) + \left( \begin{array}{cc}
0 & -I_{\mathrm{N}} \\
0 & 0
\end{array} \right) =  \left( \begin{array}{cc}
0 & -2I_{\mathrm{N}} \\
0 & 0
\end{array} \right) \in \mathfrak{a}(E).$$
\noindent Thus, $\bigl(\begin{smallmatrix}
0 & -I_{\mathrm{N}} \\
0 & 0
\end{smallmatrix}\bigr)\in \mathfrak{a}(E)$. But, by Step 2, all the $Y_{ii}$'s are in $\mathfrak{a}(E)$, so we also have :
$$\sum_{i=1}^{N} Y_{ii} =\left( \begin{array}{cc}
0 & 0 \\
I_{\mathrm{N}} & 0
\end{array} \right) \in \mathfrak{a}(E).$$
By adding these two matrices, $J=\bigl(\begin{smallmatrix}
0 & -I_{\mathrm{N}} \\
I_{\mathrm{N}} & 0
\end{smallmatrix}\bigr)\in \mathfrak{a}(E)$.
\vskip 1mm

\noindent \textbf{Step 5.} For every $i\in \aN$, $[J,Z_{ii}]=2Y_{ii}+2X_{ii}\in \mathfrak{a}(E)$. But $Y_{ii}\in \mathfrak{a}(E)$, so $2X_{ii}=[J,Z_{ii}]-2Y_{ii} \in \mathfrak{a}(E)$ and, for every $i \in \aN$, $X_{ii}\in \mathfrak{a}(E)$.
\vskip 1mm

\noindent \textbf{Step 6.} We recall that $X_{ij}=X_{ji}$ and $Y_{ij}=Y_{ji}$. Let $i\in \aN$. Substracting $(\omO-E)Y_{ii}\in \mathfrak{a}(E)$ and adding $X_{ii}\in \mathfrak{a}(E)$ in (\ref{eq_lem_a_1}) we get $Y_{i,i-1}+Y_{i,i+1} \in \mathfrak{a}(E)$. For $i=1$, it means that $Y_{1,2}\in \mathfrak{a}(E)$. Then, $\frac{1}{2}Z_{1,2}=[X_{1,1},Y_{1,2}]\in \mathfrak{a}(E)$ and $Z_{1,2}\in \mathfrak{a}(E)$. But we also have $2X_{1,2}=[Z_{1,2},X_{2,2}]\in \mathfrak{a}(E)$ and $X_{1,2}\in \mathfrak{a}(E)$. Now, for $i=2$, we have $Y_{2,1}+Y_{2,3}\in \mathfrak{a}(E)$. But we just proved that $Y_{2,1}\in \mathfrak{a}(E)$, thus $Y_{2,3}\in \mathfrak{a}(E)$. Inductively, we prove that : 
$$\forall i\in \aN,\ Y_{i,i+1}\in \mathfrak{a}(E).$$
\noindent Also, for every $i\in \aN$,
$$[X_{ii},Y_{i,i+1}]=\frac{1}{2}Z_{i,i+1} \in \mathfrak{a}(E)\quad \mathrm{and}\quad [Z_{i,i+1},X_{i+1,i+1}]=2X_{i,i+1} \in \mathfrak{a}(E).$$
\noindent It proves that all the matrices $X_{ij}$ and $Y_{ij}$ for $i,j\in \aN$ and $|i-j|\leq 1$ are in $\mathfrak{a}(E)$. Thus, $\mathfrak{a}(E)=\spN$.
\end{proof}

\noindent We can now prove Proposition \ref{prop_GE_SpN}.

\begin{proof}[Proof of Proposition \ref{prop_GE_SpN}]
In this proof we directly construct $\ell_C$ and $I(\ell,N)$. Let $\lambda_1^{\omO}$,$\ldots$, $\lambda_N^{\omO}$ be the real eigenvalues of the real symmetric matrix $M_{\omO}(0)$, as in the introduction. Then, the eigenvalues of $X_{\omO}(E) ^tX_{\omO}(E)$ are $1$, $(\lambda_1^{\omO}-E)^2$, $\ldots$, $(\lambda_N^{\omO}-E)^2$. Thus :
$$||X_{\omO}(E)||=\max \left(1, \max_{1\leq i \leq N} |\lambda_i^{\omO}-E|\right),$$
where $||\ ||$ is the matrix norm associated to the euclidian norm on $\R^{2N}$.

\noindent Let $\mathcal{O}$ be the neighborhood of the identity given by Theorem \ref{thm_breuillard} for $G=\SpN$. Then $\mathcal{O}$ depends only on $N$. To apply Theorem \ref{thm_breuillard} to $G(E)\subset \SpN$, we need to find an interval of values of $E$ such that, for $\ell$ small enough, 
\begin{equation}\label{eq_prop_GE_1}
\forall \omO\in \{ 0,1\}^N,\ 0<\ell ||X_{\omO}(E)|| < \dlO,
\end{equation}
or, equivalently,
\begin{equation}\label{eq_prop_GE_2}
0<\ell \max \left(1,\max_{\omO\in \{ 0,1\}^N} \max_{1\leq i\leq N} |\lambda_i^{\omO}-E|\right) < \dlO.
\end{equation}
We assume that $\ell \leq \dlO$ and we set $r_{\ell}=\frac{1}{\ell}\dlO \geq 1$. We want to characterize the set :
\begin{equation}\label{eq_prop_GE_3}
I(\ell,N)=\left\{E\in \R\ \bigg|\  \max \left(1,\max_{\omO\in \{ 0,1\}^N} \max_{1\leq i\leq N} |\lambda_i^{\omO}-E|\right) \leq r_{\ell} \right\}.
\end{equation}
As $r_{\ell}\geq 1$,
\begin{equation}\label{eq_prop_GE_4}
I(\ell,N)=\bigcap_{\omO\in \{ 0,1\}^N} \bigcap_{1\leq i\leq N} [\lambda_i^{\omO}-r_{\ell}, \lambda_i^{\omO}+r_{\ell}].
\end{equation}
Let $\lambda_{\mathrm{min}}$, $\lambda_{\mathrm{max}}$ and $\delta_{0}$ be as in (\ref{eq_lambda_minmax}). If $\delta_0 <r_{\ell}$, $I(\ell,N)\neq \emptyset$ and more precisely, $I(\ell,N)=[\lambda_{\mathrm{max}}-r_{\ell},\lambda_{\mathrm{min}}+r_{\ell}]$. This interval is centered in $\frac{\lambda_{\mathrm{min}}+\lambda_{\mathrm{max}}}{2}$ and is of length $2r_{\ell}-2\delta_0 >0$. Moreover, $2r_{\ell}-2\delta_0 \to +\infty$ when $\ell$ tends to $0^+$. As $\lambda_{\mathrm{min}}$, $\lambda_{\mathrm{max}}$ and $\dlO$ depend only on $N$, $I(\ell,N)$ depends only on $\ell$ and $N$ and the condition $\delta_0 <r_{\ell}$ is equivalent to
$$\ell <\frac{\dlO}{\delta_0}=\ell_C(N).$$
\noindent So, we have just proved that, 
\begin{equation}\label{eq_prop_GE_5}
\forall \ell <\ell_C,\ \forall E\in I(\ell,N),\ 0<\ell ||X_{\omO}(E)|| \leq \dlO.
\end{equation}
\noindent Thus, for every $E\in I(\ell,N)$, $\log\,T_{\omO}(E)=\ell X_{\omO}(E)$, as \emph{exp} is a diffeomorphism from $\log \mathcal{O}$ into $\mathcal{O}$. Then, we can apply Lemma \ref{lem_a_spN} to obtain :
\begin{equation}\label{eq_prop_GE_6}
\forall \ell >0,\ \forall E\in \R,\ \mathrm{Lie} \{ \ell X_{\omO}(E) \ |\ \omO \in \{0,1\}^N \}=\spN.
\end{equation}
\noindent Applying Theorem \ref{thm_breuillard}, we get :
\begin{equation}\label{eq_prop_GE_7}
\forall \ell <\ell_C,\ \forall E\in I(\ell,N),\ \overline{<T_{\omO}(E) \ |\ \omO \in \{0,1\}^N>}=\SpN.
\end{equation}
\noindent As $\overline{<T_{\omO}(E) \ |\ \omO \in \{0,1\}^N>}\subset G(E)$ and $G(E) \subset \SpN$, we finally have
\begin{equation}\label{eq_prop_GE_8}
\forall \ell <\ell_C,\ \forall E\in I(\ell,N),\ G(E)=\SpN,
\end{equation}
which proves Proposition \ref{prop_GE_SpN}.
\end{proof}

\noindent We can finally prove Theorem \ref{thm_lyap_N}.

\begin{proof}[Proof of Theorem \ref{thm_lyap_N}]
As, for every $\ell<\ell_C$ and every $E\in I(\ell,N)$, $G(E)=\SpN$, we get that, for every $\ell<\ell_C$ and every $E\in I(\ell,N)$, $G(E)$ is \LpSI. Then, by \cite[Proposition IV.3.4]{BL85}, we get the separability and positivity of the Lyapunov exponents and their integral representation (\ref{eq_lyap_rep_int}), together with the existence of the $\mu_E$-invariant measure $\nu_{p,E}$, for every $p\in \aN$. For the assertion on absence of absolutely continuous spectrum, we refer to Kotani's theory (\cite{KS88}). For the regularity result (\ref{eq_lyap_reg}), we can directly apply \cite[Theorem 2]{boumazarmp} on the interval $I(\ell,N)$.
\end{proof}

\noindent In \cite{boumazampag} we also used Theorem \ref{thm_breuillard} to obtain the separability of the Lyapunov exponents for the model studied there. This model corresponds to the case $N=2$ and $\ell=1$ of $\Hl$. The main difference between \cite{boumazampag} and the proof we have just given is that in \cite{boumazampag} we could not let $\ell$ get small and then just control $E$, to ensure that $\ell X_{\omO}(E) \in \log \mathcal{O}$ and thus $T_{\omO}(E)\in \mathcal{O}$, uniformely on $\omO \in \{0,1\}^N$. We had first used simultaneous diophantine approximation to find a suitable power of $T_{\omO}(E)$, say $(T_{\omO}(E))^{m_{\omO}(E)}$, which is in $\mathcal{O}$. Then arised difficulties with the computations of the logarithm. First, $\log(\exp(m_{\omO}(E) X_{\omO}(E)))\neq m_{\omO}(E) X_{\omO}(E)$ as $m_{\omO}(E) X_{\omO}(E) \notin \log \mathcal{O}$ in general. It leads to a problem of determination of the logarithm and the existence of a discrete set of critical energies $\mathcal{S}$ such that, for $E\in \mathcal{S}$, $\log (T_{\omO}(E))^{m_{\omO}(E)}$ is not defined. Then, for $E\notin \mathcal{S}$, the expression of these logarithms being not simple, we could not use an algebraic result like Lemma \ref{lem_a_spN} to prove that the Lie algebra generated by the logarithms is $\spN$. That is why we had to restrict ourselves to $N=2$ in \cite{boumazampag}.

\section{Estimates on the products of transfer matrices}\label{sec_estimate_trans_mat}

In this Section, we review and adapt results precising the convergence of the sequence $(\wedge^p U^{(n)}(E))_{n\in \Z}$ where, for every $n\in \Z$, $U^{(n)}(E)=T_{\omega^{(n-1)}}(E)\cdots T_{\omega^{(0)}}(E)$. In particular, we will prove large deviation type estimates. 

\noindent Let $I\subset \R$ be an open interval such that, for every $E\in I$, $G(E)$ is \LpSI.

\begin{lem}\label{lem_CVU}
$$\frac{1}{n} \E \left( \log \frac{||(\wedge^{p} U^{(n)}(E))x||}{||x||}\right) \xrightarrow[n \to +\infty]{} \gamma_{1}(E)+\ldots +\gamma_{p}(E),$$
uniformly in $E\in I$ and $\bar{x} \in \P(L_{p})$.
\end{lem}

\noindent For the proof of this lemma, we refer to \cite[Proposition V.4.9 ]{CL90} or \cite[Proposition 6.3.4]{boumazathese}.

\begin{lem}\label{lem_wegner_U}
Let $p\in \aN$. There exist $\xi_1>0$, $\delta>0$ and $n_1\in \N$ such that, for every $E\in I$, $n\geq n_1$ and $x\in \wedge^p (\R^{2N})$, $||x||=1$, we have :
\begin{equation}\label{eq_lem_wegner_U}
\E \left( ||\wedge^p U^{(n)}(E)x||^{-\delta} \right) \leq \ee^{-\xi_1 n}.
\end{equation}
\end{lem}

\begin{proof}
We fix $p\in \aN$. We set $N_U=||\wedge^p U^{(n)}(E)x||$. We start by writing 
$$N_U^{-\delta} =\ee^{-\delta \log N_U},$$
and using the inequality $\ee^y \leq 1+y+y^2\ee^{|y|}$, for any $y\in \R$. Then, for every $E\in I$ and every $\delta >0$,
\begin{equation}\label{eq_wegner_U_1}
\E \left( N_U^{-\delta} \right) \leq 1-\delta \E \left( \log N_U \right) + \delta^2 \E \left( \left(\log N_U\right)^2 \ee^{\delta \log N_U} \right).
\end{equation}
But, as $\ee^{\delta \log N_U}=N_U^{\delta}$, $||x||=1$ and the $T_{\omega^{(n)}}(E)$ are \emph{i.i.d.}, by Cauchy-Schwarz inequality,
\begin{eqnarray}
\E \left( \left(\log N_U\right)^2 \ee^{\delta \log N_U} \right) & \leq & \E\left[  \left(\log N_U\right)^4 \right]^{\frac{1}{2}} \E\left(N_U^{2\delta} \right)^{\frac{1}{2}} \nonumber \\
& \leq & \E \left[ \left(\sum_{i=0}^{n-1} p\log ||T_{\omega^{(i)}}(E)|| \right)^4 \right]^{\frac{1}{2}} \E \left[ \prod_{i=1}^{n-1} ||T_{\omega^{(i)}}(E)||^{2p\delta} \right]^{\frac{1}{2}} \nonumber \\
& \leq & n^2 p^2 \E\left[ \left(\log ||T_{\omO}(E)|| \right)^4 \right]^{\frac{1}{2}}  \E\left[||T_{\omO}(E)||^{2p\delta} \right]^{\frac{n}{2}}.\nonumber 
\end{eqnarray}  
Thus, there exist constants $C_1=C_1(I)$ and $C_2=C_2(I)$ such that, 
\begin{equation}\label{eq_wegner_U_2}
\E \left( ||\wedge^p U^{(n)}(E)x||^{-\delta} \right) \leq 1-\delta \E \left( \log ||\wedge^p U^{(n)}(E)x|| \right) + \delta^2 n^2 C_1 (C_2)^n.
\end{equation}
If $C>0$ is such that, for every $E\in I$, $T_{\omO}(E) \leq C$, we can choose $C_1=(p\log C)^2$ and $C_2=C^{p\delta}$. Moreover, by Lemma \ref{lem_CVU}, there exists $n_0\geq 1$, uniform in $E\in I$, and $x$ normalized, such that,
\begin{eqnarray}
\E \left( ||\wedge^p U^{(n_0)}(E)x||^{-\delta} \right) & \leq &  1-\frac{1}{2}n_0 \delta (\gamma_1(E)+\cdots+\gamma_p(E)) + \delta^2 n_0^2 C_1 C_2^{n_0} \nonumber \\
& \leq & 1-\varepsilon, \nonumber
\end{eqnarray}
for $\varepsilon >0$, if we choose $\delta$ small enough. Then, for $n\geq 1$, we set $[\frac{n}{n_0}]$ the largest integer less than or equal to $\frac{n}{n_0}$, and we write the euclidian division of $n$ by $n_0$, $n=[\frac{n}{n_0}]n_0 +r$, $0\leq r<n_0$. Then, there exist $n_1\geq 1$, a constant $\tilde{C}$ and $\xi_1 >0$ such that for every $p\in \aN$,
$$\forall n\geq n_1,\ \forall E\in I,\ \E \left( ||\wedge^p U^{(n_0)}(E)x||^{-\delta}\right) \leq \tilde{C} (1-\varepsilon)^{[\frac{n}{n_0}]} \leq \ee^{-\xi_1 n}.$$
For this, we refer to \cite[Lemma 5.1]{CKM87}.
\end{proof}

\noindent This lemma will be used in the proof of a Wegner estimate for $\Hl$. Later, we will need results on large deviation for the random walk $(U^{(n)}(E))_{n\in \Z}$. We will now briefly quote them, following or refering to \cite{BL85} and \cite{DSS02} for the proofs. The first result is an estimate on the $\mu_E$-invariant measure $\nu_{p,E}$ introduced at Theorem \ref{thm_lyap_N}.

\begin{prop}\label{prop_mes_inv_reg}
Let $p\in \aN$. Let $\delta(\bar{x},\bar{y})$ be the projective distance between $\bar{x}$ and $\bar{y}$ on $\P(L_p)$. We assume that $E\in \R$ is such that $G(E)$ is \LpSI. Then, there exist $\rho>0$ and $C>0$ such that, for every $\bar{x}\in \P(L_p)$ and every $\varepsilon>0$,
\begin{equation}\label{eq_mes_inv_reg}
\nu_{p,E}\left(\left\{\bar{y}\in \P(L_p)\ |\ \delta(\bar{x},\bar{y}) \leq \varepsilon \right\} \right) \leq C\varepsilon^{\rho}.
\end{equation}
\end{prop}

\begin{proof}
It comes from a simple adaptation to the symplectic case of Theorem VI.2.1 and Proposition VI.4.1 in \cite{BL85}. Thus, we deduce this result as in Corollary VI.4.2 in \cite{BL85}.
\end{proof}
\vskip 2mm

\begin{lem}\label{lem_large_deviation1}
Let $p\in \aN$. We assume that $E\in \R$ is such that $G(E)$ is \LpSI. Then, there exists $\kappa_0 >0$ such that, for every $\varepsilon >0$, $x\in L_p$, $x\neq 0$,
\begin{equation}\label{eq_lem_large_deviation1}
\limsup_{n\to +\infty} \frac{1}{\ell n} \log \mathsf{P}\left( \left| \log ||(\wedge^p U^{(n)}(E))x||-\ell n (\gamma_1+\cdots+\gamma_p)(E) \right|>\ell n \varepsilon \right) < -\kappa_0.
\end{equation}
\end{lem}

\begin{proof}
We refer to \cite[Theorem V.6.2]{BL85}, replacing $S_n$ there by $\wedge^p U^{(n)}(E)$ here and $\gamma$ by $(\gamma_1+\cdots+\gamma_p)(E)$.
\end{proof}
\vskip 2mm

\begin{lem}\label{lem_large_deviation2}
Let $p\in \aN$. We assume that $E\in \R$ is such that $G(E)$ is \LpSI. Let $y\in L_p$ of norm $||y||=1$. For every $\varepsilon >0$, there exist $\kappa_1 >0$ and $n_0\in \N$ such that, \begin{equation}\label{eq_lem_large_deviation2}
\forall n\geq n_0,\ \sup_{x\in L_p, x\neq 0} \mathsf{P} \left( \frac{|((\wedge^p U^{(n)}(E))x,y)|}{||(\wedge^p U^{(n)}(E))x||} <\ee^{-\varepsilon \ell n} \right) <\ee^{-\kappa_1 \ell n}.
\end{equation}
\end{lem}

\begin{proof}
We can directly adapt \cite[Lemma 6.3]{DSS02} by replacing $f_n$ there by $f_n :[0,\ell] \to \R$,
$$f_n(t)=\left\lbrace \begin{array}{lcl}
1 & \mathrm{if} & 0 \leq t \leq \ee^{-\varepsilon \ell n} \\
2-t\ee^{\varepsilon \ell n} & \mathrm{if} & \ee^{-\varepsilon \ell n} \leq t \leq 2\ee^{-\varepsilon \ell n} \\
0 & \mathrm{if} & 2\ee^{-\varepsilon \ell n} \leq t \leq \ell
\end{array}\right.$$
for $\varepsilon$ sufficiently small, and $\P(\R^2)$ by $\P(L_p)$ in the definition of $\Phi_n$. We also use an estimate on the convergence to $\nu_{p,E}$ of the sequence of powers of the image of $\mu_E$ by the Fourier-Laplace operator (see \cite[Corollary V.4.12]{CL90}, or \cite[Proposition 6.3.13(iii)]{boumazathese}). In particular, at (6.6) of \cite{DSS02}, we assume that $\frac{\log \rho}{\ell}+\varepsilon <0$ and we set $\delta_1=\frac{1}{2}|\frac{\log \rho}{\ell}+\varepsilon|$. At the end of the proof, we use Proposition \ref{prop_mes_inv_reg} to estimate $\nu_{p,E}$.
\end{proof}

\noindent We can finally deduce that, with probability exponentially close to $1$, the matrix elements $((\wedge^p U^{(n)}(E))x,y)$, for $p\in \aN$ and $x,y\in L_p$, grow exponentially at almost the rate of $\ell (\gamma_1+\cdots+\gamma_p)(E)$.

\begin{prop}\label{prop_large_deviation}
Let $p\in \aN$ and assume that $E\in \R$ is such that $G(E)$ is \LpSI. Let $x,y\in L_p$, $||x||=||y||=1$. Then, for every $\varepsilon >0$, there exist $\kappa >0$ and $n_0\in \N$ such that, 
\begin{equation}\label{eq_prop_large_deviation}
\forall n\geq n_0,\ \mathsf{P}\left(|((\wedge^p U^{(n)}(E))x,y)| \geq \ee^{(\gamma_1(E)+\cdots+\gamma_p(E)-\varepsilon)\ell n} \right) \geq 1-\ee^{-\kappa \ell n}.
\end{equation}
\end{prop}

\begin{proof}
Let $\varepsilon >0$ and $x,y\in L_p$, $||x||=||y||=1$. First, from Lemma \ref{lem_large_deviation1}, we deduce that there exists $n_1\in \N$ such that, for every $n\geq n_1$, 
\begin{equation}\label{eq_large_deviation1}
\mathsf{P}\left(\ee^{(\gamma_1(E)+\cdots+\gamma_p(E)-\varepsilon)\ell n} \leq ||(\wedge^p U^{(n)}(E))x|| \leq \ee^{(\gamma_1(E)+\cdots+\gamma_p(E)+\varepsilon)\ell n} \right) \geq 1-\ee^{-\kappa_0 \ell n}.
\end{equation}
Then, combining this probability estimate with (\ref{eq_lem_large_deviation2}), we get the existence of $n_2\in \N$ such that, for every $n\geq n_2$, 
{\small $$\mathsf{P}\left(|(\wedge^p U^{(n)}(E)x,y)| \geq \ee^{-\varepsilon \ell n} ||\wedge^p U^{(n)}(E)x|| \geq \ee^{(\gamma_1(E)+\cdots+\gamma_p(E)-2\varepsilon)\ell n} \right) \geq 1-\ee^{-\kappa_1 \ell n}-\ee^{-\kappa_0 \ell n}.$$}

\noindent We get (\ref{eq_prop_large_deviation}) for $n$ large enough, say $n\geq n_0$, with $\kappa=\min(\kappa_0,\kappa_1)>0$.
\end{proof}
\vskip 2mm

\noindent This result will be used in Section \ref{sec_proof_localization} to prove a probability estimate required to start a multiscale analysis.

\section{The integrated density of states}\label{sec_ids}

To prove a Wegner estimate for $\Ho$ like in Theorem \ref{thm_wegner}, a crucial property is the local H\"older continuity of the integrated density of states of $\Ho$. We review, in this Section, the definition of the integrated density of states and we prove Theorem \ref{thm_ids}.
\vskip 2mm

\noindent The integrated density of states is the distribution function of the proper energy levels, per unit volume, of $\Ho$. To define it, we consider, for every integer $L\geq 1$, the restriction $\HL$ of $\Ho$ to $L^2([-\ell L,\ell L])\otimes \C^N$, with Dirichlet boundary conditions at $\pm \ell L$.

\begin{defi}
The integrated density of states associated to $\Ho$ is the function from $\R$ to $\R_{+}$, $E\mapsto N(E)$, where $N(E)$, for $E\in \R$, is defined as the following thermodynamical limit : 
\begin{equation}\label{eq_def_ids}
N(E)=\lim_{L\to +\infty} \frac{1}{2\ell L} \# \{ \lambda \leq E |\ \lambda \in \sigma(\HL) \},
\end{equation}
for $\mathsf{P}$-almost every $\omega \in \Omega$.
\end{defi}

\noindent In this definition, appear two problems of existence. The first one is to verify that the cardinal in (\ref{eq_def_ids}) is finite for every $\omega \in \Omega$. The second one is the existence of the limit and its almost-sure independence on $\omega$. In \cite{boumazarmp}, we already proved the existence of the integrated density of states for matrix-valued continuous Schr\"odinger operators of the form (\ref{model_dimd}). In particular, the integrated densities of states of $\Ho$ and $\Hl$ are well defined for every $\ell>0$ and every $E\in \R$. Moreover, as $N(E)$ and the sum of positive Lyapunov exponents, $\gamma_{1}(E)+\cdots+\gamma_N(E)$, are harmonically conjugated through a Thouless formula (see \cite[Theorem 3]{boumazarmp}), $N(E)$ inherits the same H\"older regularity as the Lyapunov exponents. This is how we proved \cite[Theorem 4]{boumazarmp}. Applying this theorem to $\Hl$ and using Theorem \ref{thm_lyap_N}, we obtain Theorem \ref{thm_ids} as stated in the introduction. The H\"older exponent $\alpha$ in (\ref{eq_ids_reg}) is equal to the H\"older exponent of the Lyapunov exponents in (\ref{eq_lyap_reg}). This is due to properties of the Hilbert transform.

\section{A Wegner estimate}\label{sec_wegner}

This Section is devoted to the proof of Theorem \ref{thm_wegner}. For this purpose, we need two lemmas which give estimates on the norm of the solutions of the equation $-u''+Vu=0$ for $V\in L^1_{\mathrm{loc}}(\R,\mathcal{M}_{\mathrm{N}}(\R))$.

\begin{lem}\label{lem_estim1}
Let $V$ be a matrix-valued function in $L^{1}_{\mathrm{loc}}(\R,\mathcal{M}_{\mathrm{N}}(\R))$ and $u$ a solution of $-u''+Vu=0$. Then, for every $x,y\in \R$,  
\begin{equation}\label{eq_lem_estim1}
||u(x)||^{2}+||u'(x)||^{2} \leq (||u(y)||^{2}+||u'(y)||^{2})\exp \left(\int_{\min(x,y)}^{\max(x,y)} ||V(t)+1||\, \dd t \right).
\end{equation}
\end{lem}

\noindent We already proved this lemma in \cite[Lemma 2]{boumazarmp}.
\vskip 2mm

\begin{lem}\label{lem_estim2}
Let $V$ be a matrix-valued function in $L^{1}_{\mathrm{loc}}(\R,\mathcal{M}_{\mathrm{N}}(\R))$ such that, for a fixed $\ell >0$, $||V||_{\ell,\mathrm{u}}=\sup_{x\in \R} \int_{x}^{x+\ell} ||V(t)|| \dd t <\infty$. Then there exists $C>0$ such that, for every solution $u$ of $-u''+Vu=0$ and every $x\in \R$,  
\begin{equation}\label{eq_lem_estim2}
\int_{x-\ell}^{x+\ell} ||u(t)||^2 \dd t \geq C\left( ||u(x)||^2 + ||u'(x)||^2 \right).
\end{equation}
\end{lem}

\begin{proof}
Let $x\in \R$ and $u$ be a solution of $-u''+Vu=0$. Applying Lemma \ref{lem_estim1} to $x$ and $t\in[x-\ell,x+\ell]$, one gets :
\begin{eqnarray}
||u(t)||^{2}+||u'(t)||^{2} & \leq & (||u(x)||^{2}+||u'(x)||^{2})\exp \left(\int_{\min(t,x)}^{\max(t,x)} ||V(s)+1||\, \dd s \right) \nonumber \\
 & \leq & (||u(x)||^{2}+||u'(x)||^{2})\exp (2+2||V||_{\ell,\mathrm{u}}), \nonumber
\end{eqnarray}
and 
\begin{eqnarray}
||u(x)||^{2}+||u'(x)||^{2} & \leq & (||u(t)||^{2}+||u'(t)||^{2})\exp \left(\int_{\min(x,t)}^{\max(x,t)} ||V(s)+1||\, \dd s \right) \nonumber \\
 & \leq & (||u(t)||^{2}+||u'(t)||^{2})\exp(2+2||V||_{\ell,\mathrm{u}}). \nonumber
\end{eqnarray}
\noindent Setting $C_1=\exp(-2-2||V||_{\ell,\mathrm{u}})$ and $C_2=\exp(2+2||V||_{\ell,\mathrm{u}})$, we obtain, for every $t\in[x-\ell,x+\ell]$, 
\begin{equation}\label{eq_lem_estim2_1}
C_1 (||u(x)||^{2}+||u'(x)||^{2}) \leq ||u(t)||^{2}+||u'(t)||^{2} \leq C_2 (||u(x)||^{2}+||u'(x)||^{2}).
\end{equation}
We set $N_x=||u(x)||+||u'(x)||$. Using the inequality $(a+b)^2\leq 2(a^2 +b^2)$ valid for every $a,b\in \R$, we have, for every $t\in[x-\ell,x+\ell]$, 
\begin{equation}\label{eq_lem_estim2_2}
C_1 N_x^2 \leq 2 C_1 (||u(x)||^{2}+||u'(x)||^{2}) \leq 2(||u(t)||^{2}+||u'(t)||^{2}) \leq 2N_t^2,
\end{equation}
and
\begin{equation}\label{eq_lem_estim2_3}
N_t^2 \leq 2(||u(t)||^{2}+||u'(t)||^{2}) \leq 2C_2(||u(x)||^{2}+||u'(x)||^{2}) \leq 2C_2 N_x^2.
\end{equation}
Setting $C_3=\left(\frac{C_1}{2}\right)^{\frac{1}{2}}$ and $C_4=(2C_2)^{\frac{1}{2}}$, we obtain, for every $t\in[x-\ell,x+\ell]$,
\begin{equation}\label{eq_lem_estim2_4}
C_3 (||u(x)||+||u'(x)||) \leq ||u(t)||+||u'(t)|| \leq C_4(||u(x)||+||u'(x)||),
\end{equation}
which is, with our notation, $ C_3 N_x\leq N_t \leq C_4 N_x$, for every $t\in [x-\ell,x+\ell]$. 
\vskip 1mm

\noindent Assume that, for any $t\in [x-\ell,x+\ell]$, $||u(t)||< \frac{C_3}{2} N_x$. Then, $||u'(t)||> \frac{C_3}{2} N_x$ or else it would contradict (\ref{eq_lem_estim2_4}). In particular, $u'$ does not vanish on $[x-\ell,x+\ell]$ and the signs of its coordinates remain constant on this interval. Thus, for $t=x-\ell$ and $t=x+\ell$,
\begin{eqnarray}
C_3 N_x = \frac{C_3}{2} N_x + \frac{C_3}{2} N_x &  > & ||u(x+\ell)|| + ||u(x-\ell)|| \geq ||u(x+\ell)-u(x-\ell)|| \nonumber \\
  & = & \int_{x-\ell}^{x+\ell} ||u'(s)|| \dd s >2 \frac{C_3}{2} N_x = C_3 N_x. \nonumber   
\end{eqnarray}
We get a contradiction and thus, there exists $t_0\in [x-\ell,x+\ell]$, $||u(t_0)||\geq \frac{C_3}{2} N_x$. But we also have, for every $t\in [x-\ell,x+\ell]$, $||u'(t)|| \leq C_4 N_x$. Let $t\in [x-\ell,x+\ell]$ be such that $|t-t_0| \leq \frac{C_3}{4 C_4}$. Then, we have :
\begin{eqnarray}
||u(t)|| =  ||u(t_0)+u(t)-u(t_0)|| & \geq & \big| ||u(t_0)||-||u(t)-u(t_0)|| \big| \nonumber \\
& = & \left| ||u(t_0)||-\Big|\Big|\int_{t_0}^t u'(s)\dd s\Big|\Big| \right| \nonumber \\
& \geq & \frac{C_3}{2} N_x-\frac{C_3}{4} N_x = \frac{C_3}{4} N_x, \nonumber
\end{eqnarray}
because 
$$\Big|\Big|\int_{t_0}^t u'(s)\dd s\Big|\Big| \leq \int_{t_0}^t ||u'(s)||\dd s \leq \int_{t_0}^t C_4 N_x\dd s \leq \frac{C_3}{4 C_4} C_4 N_x \leq \frac{C_3}{4} N_x.$$
So, we have just proved that there exists an interval $I_0$ of length $\min \left( 2\ell,\frac{C_3}{2 C_4} \right)$, included in $[x-\ell,x+\ell]$, such that $||u(t)|| \geq \frac{C_3}{4} N_x$, for every $t\in I_0$. Then,
$$\int_{x-\ell}^{x+\ell} ||u(t)||^2 \dd t \geq \int_{I_0} ||u(t)||^2 \dd t \geq \frac{C_3^2}{16} \min\left(2\ell,\frac{C_3}{2 C_4} \right)N_x^2 \geq C(||u(x)||^2+||u'(x)||^2).$$
It proves the lemma.
\end{proof}
\vskip 2mm

\noindent We can now prove the following proposition upon which will be based the proof of Theorem \ref{thm_wegner}.

\begin{prop}\label{prop_wegner}
Let $I\subset \R$ be a compact interval and $\tilde{I}$ be an open interval, $I\subset \tilde{I}$, such that, for every $E\in \tilde{I}$, $G(E)$ is \LpSI. Then, there exist $\alpha >0$, $L_0\in \N$ and $C>0$ such that, for every $E\in I$ and every $\varepsilon>0$ :
\begin{eqnarray}
\forall L\geq L_0,& &\!\!\!\!\!\!\! \mathsf{P}\big( \big\{ \exists E'\in (E-\varepsilon,E+\varepsilon), \exists \phi \in D(\HL) \ \big|\ (\HL-E')\phi=0, \nonumber \\
\label{eq_prop_wegner}  & & ||\phi||=1\ \mathrm{and}\ ||\phi'(-\ell L)||^2+||\phi'(\ell L)||^2 \leq \varepsilon^2 \big\} \big) \leq C\,\ell \, L\, \varepsilon^{\alpha}.
\end{eqnarray}
\end{prop}
\vskip 2mm

\begin{proof}
The proof will mostly relies on the H\"older continuity of the integrated density of states of $\Ho$. Let $\tilde{I}$ be an open interval such that, for every $E\in \tilde{I}$, $G(E)$ is \LpSI. Let  $I\subset \tilde{I}$ be a compact interval. By \cite[Theorem 3]{boumazarmp}, there exist $\alpha>0$ and $C_1>0$ such that :
\begin{equation}\label{eq_prop_wegner_ids_reg}
\forall E,E'\in I,\ |N(E)-N(E')|\leq C_1 |E-E'|^{\alpha}.
\end{equation}

\noindent Let $E$ be in the interior of $I$, $\varepsilon >0$ and $L\in \N$. For every $k\in \Z$, let $I_k$ be the interval $2k\ell L+[-\ell L,\ell L]=[(2k-1)\ell L, (2k+1)\ell L]$ and denote by $H^{(I_k)}(\omega)$ the restriction of $\Ho$ to $L^2(I_k)\otimes \C^N$ with Dirichlet boundary conditions. We define the event $A_k\in \mathcal{A}$ as :
\vskip 1mm

\noindent $A_k=\{ \omega \in \Omega \ |\ H^{(I_k)}(\omega)$ has an eigenvalue $\lambda_k \in (E-\varepsilon,E+\varepsilon)$ such that the corresponding normalized eigenfunction $\phi_k$ satisfies $||\phi'(-\ell L)||^2+||\phi'(\ell L)||^2 \leq \varepsilon^2 \}$.
\vskip 1mm

\noindent Then, because the $V_{\omega}^{(n)}$ are \emph{i.i.d.} random variables, and because of the form of the potential in $\Ho$ as a $\ell$-periodization of $V_{\omega}^{(n)}$, we deduce that $\mathsf{P}(A_k)$ is independent of $k$, that is, $\forall k\in \Z$, $\mathsf{P}(A_k)=\mathsf{P}(A_0)$. Moreover, $\mathsf{P}(A_0)$ is equal to the probability in (\ref{eq_prop_wegner}).
\vskip 2mm

\noindent Let $n\in \N$ and let $J_n=\cup_{k=-n}^{n} I_k =[-(2n+1)\ell L, (2n+1)\ell L]$. Let $H^{(J_n)}(\omega)$ be the restriction of $\Ho$ to $L^2(J_n)\otimes \C^N$ with Dirichlet boundary conditions. For a fixed $\omega \in \Omega$, let $k_1,\ldots,k_j\in \{-n,\ldots,n\}$ be distinct and such that $\omega \in A_{k_i}$ for every $i\in \{1,\ldots,j\}$. Let $i\in \{1,\ldots,j\}$. Let $\phi_i$ be defined on $I_{k_i}$, $||\phi_i||=1$, $\phi_i((2k_i -1)\ell L)=\phi_i((2k_i +1)\ell L)=0$. We also assume that there exists $\lambda_{k_i} \in (E-\varepsilon,E+\varepsilon)$ such that $H^{(I_{k_i})}(\omega) \phi_i = \lambda_{k_i} \phi_i$.

\noindent Let $\chi$ be a smooth function on $\R$, $0\leq \chi \leq 1$, $\chi(x)=0$ on $(-\infty,0]$,  $\chi(x)=1$ on $[\ell,+\infty)$ and $\int_0^{\ell} \chi(x)\dd x =1$. Let $x_i^{\pm}=(2k_i \pm 1)\ell L$, so that $I_{k_i}=[x_i^-,x_i^+]$. We extend $\phi_i$ to $J_n$ by defining $\hat{\phi}_i$, for $x\in J_n$, by :
\begin{equation}\label{eq_prop_wegner_1}
\hat{\phi}_i(x)=\left\lbrace \begin{array}{lcl}
0 & \mathrm{if} &  x\notin [x_i^-,x_i^+]\\
\chi(x-x_i^-)\phi_i(x) & \mathrm{if} & x\in[x_i^-,x_i^- + \ell] \\
\phi_i(x)  & \mathrm{if} & x\in[x_i^- + \ell,x_i^+ -\ell] \\
\chi(x_i^+ -x)\phi_i(x) & \mathrm{if} & x\in[x_i^+ -\ell,x_i^+].
\end{array}\right.
\end{equation}
Then, $\hat{\phi}_i\in D(H^{(J_n)}(\omega))$ and $||\hat{\phi}_i||\leq ||\phi_i||=1$. As $H^{(I_{k_i})}(\omega) \phi_i=\lambda_{k_i}\phi_i$ and $\lambda_{k_i} \in (E-\varepsilon,E+\varepsilon)$, we have :
\begin{eqnarray}
||(H^{(J_n)}(\omega)-E)\hat{\phi}_i|| & \leq & ||(H^{(J_n)}(\omega)-\lambda_{k_i} )\hat{\phi}_i||+ ||(\lambda_{k_i}-E)\hat{\phi}_i|| \nonumber \\
 & = & ||(H^{(J_n)}(\omega)-\lambda_{k_i} )\hat{\phi}_i||+|\lambda_{k_i}-E|\, ||\hat{\phi}_i|| \nonumber \\
\label{eq_prop_wegner_2} & \leq & ||(H^{(J_n)}(\omega)-\lambda_{k_i} )\hat{\phi}_i|| + \varepsilon.
\end{eqnarray}
We want to estimate $||(H^{(J_n)}(\omega)-\lambda_{k_i} )\hat{\phi}_i||$. For every $x\in [x_i^-,x_i^- +\ell]$,
$$(\chi(x-x_i^-)\phi_i(x))''=\chi''(x-x_i^-)\phi_i(x) + 2\chi'(x-x_i^-)\phi_i '(x) + \chi(x-x_i^-)\phi_i ''(x),$$
and, for every $x\in [x_i^+-\ell,x_i^+]$,
$$(\chi(x_i^+ -x)\phi_i(x))''=\chi''(x_i^+ -x)\phi_i(x) - 2\chi'(x_i^+ -x)\phi_i '(x) + \chi(x_i^+ -x)\phi_i ''(x).$$
Thus, using $H^{(I_{k_i})}(\omega) \phi_i=\lambda_{k_i}\phi_i$,
{\small $$(H^{(J_n)}(\omega)-\lambda_{k_i})\hat{\phi}_i=\left\lbrace \begin{array}{lcl}
0 & \mathrm{if} &  x\notin [x_i^-,x_i^+]\\
-\chi''(x-x_i^-)\phi_i(x) - 2\chi'(x-x_i^-)\phi_i '(x) & \mathrm{if} & x\in[x_i^-,x_i^- + \ell] \\
0 & \mathrm{if} & x\in[x_i^- + \ell,x_i^+ -\ell] \\
-\chi''(x_i^+ -x)\phi_i(x) + 2\chi'(x_i^+ -x)\phi_i '(x) & \mathrm{if} & x\in[x_i^+ -\ell,x_i^+].
\end{array}\right.$$}
Hence we have, applying twice Lemma \ref{lem_estim1} for $x_i^-$ and for $x_i^+$ at the second inequality,
\begin{eqnarray}
||(H^{(J_n)}(\omega)-\lambda_{k_i})\hat{\phi}_i||^2 & = & \int_{x_i^-}^{x_i^- +\ell} ||\chi''(x-x_i^-)\phi_i(x) + 2\chi'(x-x_i^-)\phi_i '(x)||^2 \dd x \nonumber \\
 & & + \int_{x_i^+ -\ell}^{x_i^+} ||\chi''(x_i^+ -x)\phi_i(x) - 2\chi'(x_i^+ -x)\phi_i '(x)||^2 \dd x \nonumber \\
 & \leq & \left|\left| \left( \begin{smallmatrix}
\phi_i \\
\phi_i'
\end{smallmatrix}\right)\right|\right|_{L^{\infty}([x_i^-,x_i^- +\ell])}^2 \times \int_{x_i^-}^{x_i^- +\ell} \left|\left| \left( \begin{smallmatrix}
\chi''(x-x_i^-) \\
2\chi'(x-x_i^-)
\end{smallmatrix}\right)\right|\right|^2 \dd x \nonumber \\
 & & + \left|\left| \left( \begin{smallmatrix}
\phi_i \\
\phi_i'
\end{smallmatrix}\right)\right|\right|_{L^{\infty}([x_i^+ -\ell,x_i^+])}^2 \times \int_{x_i^+ -\ell}^{x_i^+} \left|\left| \left( \begin{smallmatrix}
\chi''(x_i^+ -x) \\
2\chi'(x_i^+ -x)
\end{smallmatrix}\right)\right|\right|^2 \dd x \nonumber \\
 & \leq & C_2 \left( \left|\left| \left( \begin{smallmatrix}
\phi_i (x_i^-) \\
\phi_i'(x_i^-) 
\end{smallmatrix}\right)\right|\right| + \left|\left| \left( \begin{smallmatrix}
\phi_i (x_i^+)\\
\phi_i'(x_i^+)
\end{smallmatrix}\right)\right|\right| \right)^2 \nonumber \\
\label{eq_prop_wegner_3} & = & C_2 \left( ||\phi_i'(x_i^-)||^2+||\phi_i'(x_i^+)||^2 \right) \leq C_2 \varepsilon^2,
\end{eqnarray}
using the fact that $\omega \in A_{k_i}$ and using the Dirichlet boundary conditions of $H^{(I_{k_i})}(\omega)$ at $x_i^-$ and $x_i^+$ to say that $\phi_i(x_i^-)=\phi_i(x_i^+)=0$. The constant $C_2$ depends only on $\chi$ and the parameters of the potential of $\Ho$. We normalize $\hat{\phi}_i$ by setting $\tilde{\phi}_i=\hat{\phi}_i/||\hat{\phi}_i||$. We also have $||\hat{\phi}_i|| \geq \frac{1}{2}$ because $||\phi_i||=1$ and $\int_0^{\ell} \chi(x) \dd x=1$, and thus, by (\ref{eq_prop_wegner_2}) and (\ref{eq_prop_wegner_3}),
\begin{eqnarray}
||(H^{(J_n)}(\omega)-E)\tilde{\phi}_i|| & = & ||\hat{\phi}_i||^{-1} ||(H^{(J_n)}(\omega)-E)\hat{\phi}_i|| \nonumber \\
\label{eq_prop_wegner_4} & \leq & 2 \sqrt{C_2}\varepsilon := C_3 \varepsilon.
\end{eqnarray}

\noindent We have construct, for each $i\in \{1,\ldots,j \}$, a normalized function $\tilde{\phi}_i$ in $D(H^{(J_n)}(\omega))$, supported in $I_{k_i}$, such that :
\begin{equation}\label{eq_prop_wegner_5}
\forall i\in \{1,\ldots,j \},\ ||(H^{(J_n)}(\omega)-E)\tilde{\phi}_i||\leq C_3 \varepsilon,
\end{equation}
where $C_3$ depends only on the choice of $\chi$ and on the parameters of the potential of $\Ho$. Moreover, as $\tilde{\phi}_i$ is supported in $I_{k_i}$ and the intervals $I_{k_1},\ldots,I_{k_j}$ are disjoints, $(\tilde{\phi}_1,\ldots,\tilde{\phi}_j)$ is an orthonormal set and :
\begin{equation}\label{eq_prop_wegner_6}
\forall i\neq i',\ (\tilde{\phi}_i,H^{(J_n)}(\omega)\tilde{\phi}_{i'})=0=(H^{(J_n)}(\omega)\tilde{\phi}_{i}, H^{(J_n)}(\omega)\tilde{\phi}_{i'}).
\end{equation}
\vskip 2mm

\noindent We recall that, as proven in \cite[Section 2.3]{boumazarmp}, the spectrum of $H^{(J_n)}(\omega)$ is a discrete set of eigenvalues with only $+\infty$ as accumulation point, and thus, its number of eigenvalues in any compact interval is finite. As we have (\ref{eq_prop_wegner_5}) and (\ref{eq_prop_wegner_6}), we can apply to $(\tilde{\phi}_1,\ldots,\tilde{\phi}_j)$ and $H^{(J_n)}(\omega)$ the version of Temple's inequality given in \cite[Lemma A.3.2]{ST85} to obtain that the number of eigenvalues of $H^{(J_n)}(\omega)$ in $[E-C_3\varepsilon,E+C_3\varepsilon]$, counted with multiplicity, is at least $j$. So we have, for a fixed $\omega \in \Omega$,
$$j=\#\big\{ k\in\{-n,\ldots,n\}\ \big|\ \omega \in A_k \big\} \leq \# \big\{ \lambda \in [E-C_3\varepsilon,E+C_3\varepsilon]\ \big|\ \lambda \in\sigma_{\mathrm{p}}(H^{(J_n)}(\omega)) \big\}. $$
Moreover, applying the law of large numbers to the random variables $\mathbf{1}_{A_{-n}},\ldots,\mathbf{1}_{A_n}$, we get that, for $\mathsf{P}$-almost every $\omega \in \Omega$,
$$\frac{1}{2n+1} \#\big\{ k\in\{-n,\ldots,n\}\ \big|\ \omega \in A_k \big\}=\frac{1}{2n+1} \left(\mathbf{1}_{A_{-n}}+\ldots+\mathbf{1}_{A_n} \right) \xrightarrow[n\to +\infty]{} \E(\mathbf{1}_{A_0}),$$
with $\E(\mathbf{1}_{A_0})=\mathsf{P}(A_0)$. Now, we assume that $\varepsilon$ is small enough to ensure that $[E-C_3\varepsilon,E+C_3\varepsilon]\subset I\subset \tilde{I}$ and to apply (\ref{eq_prop_wegner_ids_reg}) on $[E-C_3\varepsilon,E+C_3\varepsilon]$. Then we have, for $\mathsf{P}$-almost every $\omega \in \Omega$,
\begin{eqnarray}
\mathsf{P}(A_0) & = & \lim_{n \to +\infty} \frac{1}{2n+1} \#\big\{ k\in\{-n,\ldots,n\}\ \big|\ \omega \in A_k \big\} \nonumber \\
 & \leq & \lim_{n \to +\infty} \frac{1}{2n+1} \# \big\{ \lambda \in [E-C_3\varepsilon,E+C_3\varepsilon]\ \big|\ \lambda \in\sigma_{\mathrm{p}}(H^{(J_n)}(\omega)) \big\} \nonumber \\
 & = & 2\ell L  \lim_{n \to +\infty} \frac{1}{2(2n+1)\ell L} \# \big\{ \lambda \in [E-C_3\varepsilon,E+C_3\varepsilon]\ \big|\ \lambda \in\sigma_{\mathrm{p}}(H^{(J_n)}(\omega)) \big\} \nonumber \\
 & = & 2\ell L (N(E+C_3\varepsilon)-N(E-C_3\varepsilon)) \nonumber \\
 & \leq & 2\ell L C_1 (2C_3 \varepsilon)^{\alpha} :=C \ell L \varepsilon^{\alpha}. \nonumber 
\end{eqnarray}
It finishes the proof.
\end{proof}

\noindent We remark that the exponent $\alpha$ in (\ref{eq_prop_wegner}) is the same as the H\"older exponent of the Lyapounov exponent and the integrated density of states. We can now use Proposition \ref{prop_wegner}, Lemma \ref{lem_wegner_U} and Lemma \ref{lem_estim2} to prove Theorem \ref{thm_wegner}.

\begin{proof}[Proof of Theorem \ref{thm_wegner}]
Let $I\subset \R$ be a compact interval and $\tilde{I}$ be an open interval, $I\subset \tilde{I}$, such that, for every $E\in \tilde{I}$, $G(E)$ is \LpSI. Let $\beta \in (0,1)$ and $\kappa >0$. For $L\in \N$, we set $n_L=[\tau (\ell L)^{\beta}]+1$ with some arbitrary $\tau >0$, where $[\tau (\ell L)^{\beta}]$ is the largest integer less or equal to $\tau (\ell L)^{\beta}$. For every $E\in I$ and $\theta_0 >0$, we define the events :
\begin{equation}\label{eq_thm_wegner_1}
A_{\theta_0}^{(L)}(E)=\big\{ \omega \in \Omega\ |\ ||T_{\omega^{(n_L -L-1)}}(E)\ldots T_{\omega^{(-L)}}(E) \left( \begin{smallmatrix}
1 \\
0
\end{smallmatrix}\right) || > \ee^{\theta_0 (\ell L)^{\beta}} \big\},
\end{equation}

\begin{equation}\label{eq_thm_wegner_2}
B_{\theta_0}^{(L)}(E)=\big\{ \omega \in \Omega\ |\ ||T_{\omega^{(L+1-n_L)}}(E)\ldots T_{\omega^{(L)}}(E) \left( \begin{smallmatrix}
0 \\
1
\end{smallmatrix}\right) || > \ee^{\theta_0 (\ell L)^{\beta}} \big\}.
\end{equation}
Let $\xi_1 >0$ and $\delta>0$ be the constants given by Lemma \ref{lem_wegner_U}. Let $\theta=\frac{\tau \xi_1}{2\delta}$ and let $C^{(L)}(E)\in \mathcal{A}$ be the event :
$$\Big\{\omega\in \Omega\ \Big|\ d\left(E, \sigma(\HL)\right) \leq \ee^{-\kappa (\ell L)^{\beta}} \Big\}.$$
If we set :
\begin{eqnarray}
 (a)& := & \mathsf{P}\left( C^{(L)}(E) \cap \bigcap_{\{E'\ |\ |E-E'|\leq \ee^{-\kappa (\ell L)^{\beta}} \} } \left[ A_{\frac{\theta}{2}}^{(L)}(E') \cap B_{\frac{\theta}{2}}^{(L)}(E') \right] \right),  \nonumber \\
 (b)& := & \mathsf{P}\left( A_{\theta}^{(L)}(E)\cap B_{\theta}^{(L)}(E) \cap \bigcup_{\{E'\ |\ |E-E'|\leq \ee^{-\kappa (\ell L)^{\beta}} \} } \left( A_{\frac{\theta}{2}}^{(L)}(E') \right)^c \right),\nonumber \\
 (c)& := & \mathsf{P}\left( A_{\theta}^{(L)}(E)\cap B_{\theta}^{(L)}(E) \cap \bigcup_{\{E'\ |\ |E-E'|\leq \ee^{-\kappa (\ell L)^{\beta}} \} } \left( B_{\frac{\theta}{2}}^{(L)}(E') \right)^c \right), \nonumber\\
 (d)& := & \mathsf{P}\left( \left(A_{\theta}^{(L)}(E) \right)^c \right) + \mathsf{P}\left( \left(B_{\theta}^{(L)}(E) \right)^c \right), \nonumber
\end{eqnarray}
then we have :
\begin{equation}\label{eq_thm_wegner_7} 
\mathsf{P}\left( \Big\{ \omega\in \Omega\ \Big|\ d\left(E, \sigma(\HL)\right) \leq \ee^{-\kappa (\ell L)^{\beta}} \Big\} \right) \leq (a)+(b)+(c)+(d).
\end{equation}
Using Tchebychev's inequality and Lemma \ref{lem_wegner_U}, applied for $p=1$, we directly get, for $L$ large enough,
\begin{equation}\label{eq_thm_wegner_8} 
(d) \leq 2\ee^{-\xi_1 n_L -\delta \theta (\ell L)^{\beta}} \leq 2\ee^{-\xi_1 \tau (\ell L)^{\beta} + \frac{\tau \xi_1}{2} (\ell L)^{\beta}}=2\ee^{-\frac{\tau \xi_1}{2} (\ell L)^{\beta}}.
\end{equation}
To estimate $(b)+(c)$, we use the fact that there exists a constant $C_0>0$ independent of $n,\omega,E$ such that, for every $E,E'\in I$,
\begin{equation}\label{eq_thm_wegner_9} 
\forall n\in \Z,\ ||T_{\omega^{(n)}}(E)-T_{\omega^{(n)}}(E')|| \leq C_0 |E-E'|.
\end{equation}
It was proven in the proof of \cite[Theorem 2]{boumazarmp}. From this, we deduce that the event $\left( A_{\frac{\theta}{2}}^{(L)}(E') \right)^c \cap A_{\theta}^{(L)}(E)$ occurs for at least one $E'$ such that $|E-E'|\leq \ee^{-\kappa (\ell L)^{\beta}}$. Then, following \cite{CKM87}, we prove that for this $E'$, there exists $\alpha_0 >0$ such that, if $\tau >0$ is small enough,
\begin{equation}\label{eq_thm_wegner_10} 
\mathsf{P}\left( A_{\theta}^{(L)}(E) \cap \left( A_{\frac{\theta}{2}}^{(L)}(E') \right)^c \right) \leq \ee^{-\alpha_0 (\ell L)^{\beta}}.
\end{equation}
We have a similar inequality for $B_{\theta}^{(L)}(E) \cap \left( B_{\frac{\theta}{2}}^{(L)}(E'') \right)^c$ for at least one $E''$ such that $|E-E''|\leq \ee^{-\kappa (\ell L)^{\beta}}$. Thus, by inclusions of the events :
\begin{eqnarray}
\label{eq_thm_wegner_11}\qquad (b)+(c) & \leq &  \mathsf{P}\left( A_{\theta}^{(L)}(E) \cap \left( A_{\frac{\theta}{2}}^{(L)}(E') \right)^c \right) + \mathsf{P}\left( B_{\theta}^{(L)}(E) \cap \left( B_{\frac{\theta}{2}}^{(L)}(E'') \right)^c \right) \\
& \leq & 2\ee^{-\alpha_0 (\ell L)^{\beta}}. \nonumber
\end{eqnarray}
It remains to estimate $(a)$. Let $\omega$ be in the event in the probability $(a)$. Let $E'\in (E-\ee^{-\kappa(\ell L)^{\beta}},E+\ee^{-\kappa(\ell L)^{\beta}})$ be an eigenvalue of $\HL$ with a normalized eigenvector $\phi$. As $\omega$ is in the event in the probability $(a)$, we have, using Lemma \ref{lem_estim2},
\begin{equation}\label{eq_thm_wegner_12} 
||\phi(-\ell L)||^2 + ||\phi'(-\ell L)||^2 = ||\phi'(-\ell L)||^2 \leq 2\ee^{-\theta (\ell L)^{\beta}}, 
\end{equation}
and
\begin{equation}\label{eq_thm_wegner_13} 
||\phi(\ell L)||^2 + ||\phi'(\ell L)||^2 = ||\phi'(\ell L)||^2 \leq 2\ee^{-\theta (\ell L)^{\beta}}.
\end{equation}
Now, using Proposition \ref{prop_wegner} with $\varepsilon=\ee^{-\kappa(\ell L)^{\beta}}$, we get :
\begin{equation}\label{eq_thm_wegner_14} 
(a) \leq C\ell L \max\left( \ee^{-\kappa(\ell L)^{\beta}}, 2\sqrt{2} \ee^{-\frac{\theta}{2} (\ell L)^{\beta}} \right)^{\alpha}.
\end{equation}
Putting (\ref{eq_thm_wegner_8}), (\ref{eq_thm_wegner_11}) and (\ref{eq_thm_wegner_14}) in (\ref{eq_thm_wegner_7}), we finally obtain (\ref{eq_thm_wegner}) for a suitable $\xi >0$ and $L$ large enough.
\end{proof}

\section{Localization properties for $\Ho$ and $\Hl$}\label{sec_localization}

In this Section, we will prove Theorem \ref{thm_localization} and its corollary, Theorem \ref{thm_localization_Hl}. It will be the content of Section \ref{sec_proof_localization}. Before that, we will present in Section \ref{sec_MSA}, the requirements needed to perform a multiscale analysis.

\subsection{Requirements of the multiscale analysis}\label{sec_MSA}

\noindent In this Section we present the properties of $\Ho$ needed to use the multiscale analysis. These properties are, for most of them, already detailed in \cite{stollmann} and \cite{DS01}, but we will follow here the notations of \cite{K07}, based upon \cite{GK01}, as we did in the introduction for the definitions of spectral and dynamical localization.

\noindent We start by giving a property that guarantees the existence of a generalized eigenfunction expansion for $\Ho$. If we denote by $\mathcal{H}$ the Hilbert space $L^2(\R)\otimes \C^N$, given $\nu > \frac{1}{4}$, we define the weighted spaces $\mathcal{H}_{\pm}$ by :
$$\mathcal{H}_{\pm}=L^2(\R,<x>^{\pm 4\nu} \dd x)\otimes \C^N,$$
where $<x>$ is as in the introduction, equal to $\sqrt{1+|x|^2}$, for any $x\in \R$. We define on $\mathcal{H}_+ \times \mathcal{H}_-$ the sesquilinear form $<\ ,\ >_{\mathcal{H}_+ ,  \mathcal{H}_-}$ by :
$$\forall (\phi,\psi) \in \mathcal{H}_+ \times \mathcal{H}_- ,\ <\phi,\psi>_{\mathcal{H}_+ ,  \mathcal{H}_-}=\int_{\R} {^t}\phi(x) \overline{\psi(x)} \dd x.$$
We also set $T$ to be the self-adjoint operator on $\mathcal{H}$ given by the multiplication by $<x>^{2\nu}$. We recall that $E_{\omega}(.)$ denotes the spectral projection of $\Ho$ and we present a property of Strong Generalized Eigenfunction Expansion.

\begin{defi}\label{def_SGEE}
Let $I\subset \R$ be an open interval. We say that $\Ho$ has the property (SGEE) on $I$ if, for some $\nu >\frac{1}{4}$,
\begin{itemize}
\item[(i)] for $\mathsf{P}$-almost every $\omega\in \Omega$, the set $\mathcal{D}_+ (\omega)=\{ \phi \in D(\Ho) \cap \mathcal{H}_+ \ |\ \Ho\phi \in \mathcal{H}_+ \}$ is dense in $\mathcal{H}_+$ and is an operator core for $\Ho$,
\item[(ii)] there exists a bounded, continuous function $f$ on $\R$, strictly positive on $\sigma(\Ho)$ such that :
$$\E \left( \left( \mathrm{tr}_{\mathcal{H}} (T^{-1} f(\Ho) E_{\omega}(I) T^{-1})\right)^2\right) <\infty.$$
\end{itemize} 
\end{defi}

\noindent Now, we can give the definition of a generalized eigenfunction and of a generalized eigenvalue.

\begin{defi}\label{def_Gef_Gev}
A measurable function $\psi:\R\to \C^N$ is said to be a generalized eigenfunction of $\Ho$ with generalized eigenvalue $\lambda$ if $\psi \in \mathcal{H}_- \setminus \{0\}$ and :
$$\forall \phi \in \mathcal{D}_+ (\omega),\ <H(\omega)\phi,\psi>_{\mathcal{H}_+ , \mathcal{H}_-} = \overline{\lambda}<\phi,\psi>_{\mathcal{H}_+ , \mathcal{H}_-}.$$
\end{defi}

\noindent We now introduce definitions and notations for the restrictions of $\Ho$ to intervals of $\R$ of finite length. For $x\in \Z$ and $L\geq 1$, we denote by $I_{L}(x)$ the interval $I_{L}(x)=[x-\ell L,x+\ell L]$, centered at $x$ and of length $2\ell L$. As in the introduction, we denote by $\mathbf{1}_{x,L}$ the characteristic function of $I_{L}(x)$ and simply by $\mathbf{1}_x$, the characteristic function of $I_1(x)$. For $L\in 3\N^*$, we also set,
$$\mathbf{1}_{x,L}^{\mathrm{out}}= \mathbf{1}_{x,L}-\mathbf{1}_{x,L-2}\quad \mathrm{and}\quad \mathbf{1}_{x,L}^{\mathrm{in}}= \mathbf{1}_{x,\frac{L}{3}}.$$

\noindent For every $x\in \Z$ and every $L\geq 1$, we denote by $H^{(x,L)}(\omega)$ the restriction of $\Ho$ to $L^2(I_{L}(x))\otimes \C^N$ with Dirichlet boundary conditions, and, for $E\notin \sigma(H^{(x,L)}(\omega))$, by $R^{(x,L)}(E)$ the resolvent of $H^{(x,L)}(\omega)$ at $E$, $R^{(x,L)}(E)=(H^{(x,L)}(\omega)-E)^{-1}$. We also denote by $E_{\omega}^{(x,L)}$ the spectral projection of $H^{(x,L)}(\omega)$. With all these notations, we can state the following Simon-Lieb type inequality property.

\begin{defi}\label{def_SLI}
Let $I\subset \R$ be a compact interval. We say that $\Ho$ has the property (SLI) if there exists a constant $C_I$ such that, given $L,L',L''\in \N$ and $x,y,y'\in \Z$, with $I_{L''}(y) \subset I_{L'-2}(y') \subset I_{L-2}(x)$, for $\mathsf{P}$-almost every $\omega \in \Omega$, if $E\in I$, $E\notin \sigma(H^{(x,L)}(\omega)) \cup \sigma(H^{(y',L')}(\omega))$, we have :
$$||\mathbf{1}_{x,L}^{\mathrm{out}} R^{(x,L)}(E) \mathbf{1}_{y,L''}|| \leq C_{I} ||\mathbf{1}_{y',L'}^{\mathrm{out}} R^{(y',L')}(E) \mathbf{1}_{y,L''}||\, ||\mathbf{1}_{x,L}^{\mathrm{out}} R^{(x,L)}(E) \mathbf{1}_{y',L'}^{\mathrm{out}}||.$$
\end{defi}
 
\noindent The property (SLI) is an estimate of how the finite length resolvents $R^{(x,L)}(E)$ vary in norms when we go from one interval to a larger one containing the first one. It is also called a Geometric Resolvent Inequality in \cite{stollmann}. We now state a property which is an estimate of generalized eigenfunctions in terms of finite length resolvents. It is called an Eigenfunction Decay Inequality.

\begin{defi}\label{def_EDI}
Let $I\subset \R$ be a compact interval. We say that $\Ho$ has the property (EDI) if there exists a constant $\tilde{C}_I$ such that, for $\mathsf{P}$-almost every $\omega \in \Omega$, given a generalized eigenvalue $E\in I$, we have for any $x\in \Z$ and any $L\in \N$ with $E\notin \sigma(H^{(x,L)}(\omega))$,
$$||\mathbf{1}_{x} \psi|| \leq \tilde{C}_I ||\mathbf{1}_{x,L}^{\mathrm{out}} R^{(x,L)}(E) \mathbf{1}_{x}||\, ||\mathbf{1}_{x,L}^{\mathrm{out}} \psi||.$$
\end{defi}

\noindent The next property is an estimate of the average number of eigenvalues of $H^{(x,L)}(\omega)$.

\begin{defi}\label{def_NE}
Let $I\subset \R$ be a compact interval. We say that $\Ho$ has the property (NE) if there exists a finite constant $\hat{C}_I$ such that, for every $x\in \Z$ and $L\in \N$,
$$\E \left( \mathrm{tr}_{\mathcal{H}} (E_{\omega}^{(x,L)}(I))\right) \leq \hat{C}_I \ell L.$$
\end{defi}

\noindent The last property required for the multiscale analysis is of a different nature. It is a probabilistic property of independence of distant intervals. An event $A\in \mathcal{A}$ is said to be \emph{based on $I_{L}(x)$} if it is determined by conditions on $H^{(x,L)}(\omega)$. Given $d_0>0$, we say that $I_{L}(x)$ and $I_{L'}(x')$ are \emph{$d_0$-nonoverlapping} if $d(I_{L}(x),I_{L'}(x'))>d_0$.

\begin{defi}\label{def_IAD}
We say that $\Ho$ has the property (IAD) if there exists $d_0>0$ such that events based on $d_0$-nonoverlapping intervals are independent.
\end{defi}

\noindent Before giving the definition of the multiscale analysis set $\Sigma_{\mathrm{MSA}}$, we need a last definition.

\begin{defi}\label{def_goodbox}
Let $\gamma,E\in \R$ and $\omega \in \Omega$. For $x\in \Z$ and $L\in 3\N^*$, we say that the interval $I_{L}(x)$ is $(\omega,\gamma,E)$-good if $E\notin \sigma(H^{(x,L)}(\omega))$ and $$||\mathbf{1}_{x,L}^{\mathrm{out}} R^{(x,L)}(E) \mathbf{1}_{x,L}^{\mathrm{in}}|| \leq \ee^{-\gamma \ell \frac{L}{3}}.$$  
\end{defi}

\noindent We can now define the multiscale analysis set. We assume that $\Ho$ has the property (IAD).

\begin{defi}\label{def_MSA}
The set $\Sigma_{\mathrm{MSA}}$ for $\Ho$ is the set of $E\in \Sigma$ for which there exists an open interval $I$ such that $E\in I$ and, given any $\zeta$, $0<\zeta <1$, and $\alpha_0 \in (1,\zeta^{-1})$, there is a length scale $L_0\in 6\N$ and a real number $\gamma >0$, so if we set $L_{k+1}=\max\{L\in 6\N\ |\ L\leq L_k^{\alpha_0}\}$ for every $k\in \N$, we have :
$$\mathsf{P}\left(\big\{ \omega \in \Omega\ |\ \forall E'\in I,\ I_{L}(x)\ \mathrm{or}\ I_L(y)\ \mathrm{is}\ (\omega,\gamma,E')-\mathrm{good} \big\} \right) \geq 1-\ee^{-L_k^{\zeta}}.$$
for every $k\in \N$ and $x,y\in \Z$ with $|x-y|>L_k +d_0$.
\end{defi}

\noindent We finish this Section by stating the bootstrap multiscale analysis theorem of \cite[Theorem 3.4]{GK01} for operators involving singular probability measure like $\Ho$.

\begin{thm}[\cite{GK01}, Theorem 3.4]\label{thm_MSA_GK}
Assume that $\Ho$ has the properties (IAD), (SLI), (NE) and verify a Wegner estimate (W) like (\ref{eq_thm_wegner}) on an open interval $I\subset \R$. Given $\gamma >0$, for each $E\in I$, there exists an integer $L_{\gamma}(E)$, bounded on compact subintervals of $I$, such that, if for a given $E_0\in \Sigma \cap I$ we have :
\begin{equation}\label{eq_thm_MSA}
\mathsf{P}\left(\big\{ \omega \in \Omega\ |\ I_{L_0}(0)\ \mathrm{is}\ (\omega,\gamma,E_0)-\mathrm{good} \big\} \right) \geq 1-\ee^{-\delta \ell L},
\end{equation}
for $L_0\in \N$, $L_0>L_{\gamma}(E)$ and $\delta >0$, then $E_0 \in \Sigma_{\mathrm{MSA}}$. 
\end{thm}

\noindent The assumption (\ref{eq_thm_MSA}) is also known as an Initial Length Scale Estimate (ILSE) in \cite{DSS02} and essentially, it remains to prove such an (ILSE) for our operator $\Ho$ on a valid interval, to prove localization on this interval. It is the main purpose of the next Section.

\subsection{Proof of the localization for $\Ho$ and $\Hl$}\label{sec_proof_localization}

To prove theorems \ref{thm_localization} and \ref{thm_localization_Hl}, we have to establish a link between multiscale analysis and the properties (EL), (SDL) and (SSEHSKD) defined in the introduction. This link is established in the following theorem.

\begin{thm}[\cite{K07}, Theorem 6.1]\label{thm_loc_MSA}
Let $I\subset \R$ be an open interval on which $\Ho$ has the properties (IAD), (SGEE) and (EDI). Then :
$$\Sigma_{\mathrm{MSA}} \cap I \subset \Sigma_{\mathrm{EL}} \cap \Sigma_{\mathrm{SSEHSKD}} \cap I \subset \Sigma_{\mathrm{EL}} \cap \Sigma_{\mathrm{SDL}} \cap I.$$
\end{thm}

\noindent According to this theorem, to prove theorems \ref{thm_localization} and \ref{thm_localization_Hl}, it only remains to prove an (ILSE) for $\Ho$ to be able to apply theorem \ref{thm_MSA_GK} for every energies on a suitable interval. We can summarize in the following figure, the ingredients of a proof of localization using multiscale analysis.

\begin{equation}\label{fig_localization_plan}
\underbrace{\mathrm{(IAD)}+\mathrm{(SLI)}+\mathrm{(NE)}+\mathrm{(W)}+\mathrm{(ILSE)}}_{\Downarrow}
\end{equation} 
$$\qquad \qquad \qquad \qquad \ \underbrace{\mathrm{(MSA)}+\mathrm{(SGEE)}+\mathrm{(EDI)}}_{\Downarrow} $$
$$\qquad \qquad \qquad \qquad \ \overbrace{\mathrm{(EL)}+\mathrm{(SDL)}+\mathrm{(SSEHSKD)}}$$
\vskip 6mm

\begin{prop}\label{prop_ILSE}
Let $I\subset \R$ be an open interval such that, for every $E\in I$, $G(E)$ is \LpSI. Let $E\in I$. For every $\varepsilon >0$, there exist $\delta >0$ and $L_0 \in \N$ such that, for every $L\geq L_0$, $L\in 3\N^*$,
\begin{equation}\label{eq_prop_ILSE}
\mathsf{P}\left( \big\{ I_L(0)\ \mathrm{is}\ (\omega,\gamma_1(E)-\varepsilon,E)-\mathrm{good} \big\} \right) \geq 1-\ee^{-\delta \ell L}.
\end{equation}
\end{prop}
\vskip 2mm

\begin{proof}
We fix $E\in I$ and assume that $L\in 3\N^*$. We consider $U_+$ and $U_-$ two matrices in $\mathcal{M}_{\mathrm{N}}(\R)$, solutions of $\Ho U_{\pm}=EU_{\pm}$ and such that :
\begin{equation}\label{eq_proof_ILSE_0}
U_+(\ell L)=U_-(-\ell L)=0\quad \mathrm{and}\quad U_{+}^{'}(\ell L)=U_{-}^{'}(-\ell L)=I_{\mathrm{N}}.
\end{equation}
Let $W(U_+,U_-)$ denote the matrix-valued Wronskian of $U_+$ and $U_-$ defined by :
\begin{equation}\label{eq_proof_ILSE_1}
\forall x\in I_{L}(0),\ W(U_+,U_-)(x)={^t}U_{-}^{'}(x)U_+(x) -{^t}U_-(x)U_{+}^{'}(x). 
\end{equation}
From \cite[Proposition III.5.5]{CL90}, $W(U_+,U_-)$ is constant on $I_L(0)$ and it is non-invertible if and only if $E$ is an eigenvalue of $\HL$. We recall that the spectrum of $\HL$ consists on a discrete set of eigenvalues of $\HL$ with only $+\infty$ as accumulation point. Thus, $E\in \sigma(\HL)$ if and only if $W(U_+,U_-)$ is non-invertible in $\mathcal{M}_{\mathrm{N}}(\R)$. By \cite[Proposition III.5.6]{CL90}, the Green kernel of $\HL$ is given by :
\begin{equation}\label{eq_proof_ILSE_2}
\forall E\notin \sigma(\HL),\ G^{(L)}(E,x,y)= \left\lbrace \begin{array}{lcl}
U_-(x)W(U_-,U_+)^{-1}\ {^t}U_+(y) & \mathrm{if} & x\leq y \\
U_+(x)W(U_+,U_-)^{-1}\ {^t}U_-(y) & \mathrm{if} & x>y.
\end{array}\right.
\end{equation}
To estimate the norm of $R^{(0,L)}(E)$, we can estimate the norm of its kernel, the Green kernel  $G^{(L)}(E,x,y)$. We start by estimating the Wronskian. As it is constant, we have, using (\ref{eq_proof_ILSE_1}),
\begin{equation}\label{eq_proof_ILSE_3}
W(U_+,U_-)=W(U_+,U_-)(\ell L)=-{^t}U_-(\ell L),
\end{equation}
and $||W(U_+,U_-)||=||U_-(\ell L)||$. But, if $X,Y\in \mathcal{M}_{\mathrm{2N}\times \mathrm{N}}(\R)$, applying Proposition \ref{prop_large_deviation} for $p=1$, column by column, we have :
\begin{equation}\label{eq_proof_ILSE_4}
\exists n_0 \geq 1,\ \forall n\geq n_0,\ \mathsf{P}\left( ||{^t}Y U^{(n)}(E) X|| \geq \ee^{(\gamma_1(E)-\varepsilon) \ell n} \right) \geq 1-\ee^{-\kappa \ell n},
\end{equation}
with $\varepsilon >0$ and $\kappa>0$ as in the proposition. If $T_{-\ell L}^{\ell L}(E)$ denote the transfer matrix from $-\ell L$ to $\ell L$, we have $U_-(\ell L)=(I_{\mathrm{N}},0)T_{-\ell L}^{\ell L}(E) {^t}(I_{\mathrm{N}},0)$. Thus, applying (\ref{eq_proof_ILSE_4}) for $Y={^t}(I_{\mathrm{N}},0)$ and $X={^t}(I_{\mathrm{N}},0)$, and using the fact that the transfer matrices are \emph{i.i.d.},
\begin{equation}\label{eq_proof_ILSE_5}
\exists L_1 \geq 1,\ \forall L\geq L_1,\ \mathsf{P}\left( ||W(U_+,U_-)|| \geq \ee^{2(\gamma_1(E)-\varepsilon) \ell L} \right) \geq 1-\ee^{-2\kappa \ell L}.
\end{equation}
Let $x\in [\ell L-\ell,\ell L]$ and $y\in [-\ell\frac{L}{3},\ell\frac{L}{3}]$. Then $x>y$ and for $E\notin \sigma(\HL)$, $G^{(L)}(E,x,y)=U_+(x)W(U_+,U_-)^{-1}\ {^t}U_-(y)$. We apply Lemma \ref{lem_estim1} to $U_+$, for $x$ and $y=\ell L$, and using (\ref{eq_proof_ILSE_0}) :
\begin{equation}\label{eq_proof_ILSE_6}
||U_+(x)|| \leq C_1,
\end{equation}
with $C_1$ independent of $\omega$ and $L$.
\vskip 2mm

\noindent To estimate the norm of ${^t}U_-(y)$ is more complicated. We start by writing :
$$||{^t}U_-(y)|| \leq \left|\left| \left( \begin{smallmatrix}
U_{-}(y) \\
U_{-}^{'}(y)
\end{smallmatrix} \right) \right|\right| = \left|\left| T_{y}^{-\ell L}(E)\left( \begin{smallmatrix}
0 \\
I_{\mathrm{N}}
\end{smallmatrix} \right) \right|\right| \leq ||T_{y}^{[y+\ell]-\ell}(E)||\, \left|\left| T_{[y+\ell]-\ell}^{-\ell L}(E)\left( \begin{smallmatrix}
0 \\
I_{\mathrm{N}}
\end{smallmatrix} \right) \right|\right|.$$
By Lemma \ref{lem_estim1} for $y$ and $-\ell L$ and using (\ref{eq_proof_ILSE_0}), we have
\begin{equation}\label{eq_proof_ILSE_7}
||T_{y}^{[y+\ell]-\ell}(E)|| \leq C_2,
\end{equation}
with $C_2$ independent of $\omega$ and $L$. Now, using the \emph{i.i.d.} character of the transfer matrices and (\ref{eq_large_deviation1}) for $p=1$, we get the existence of $L_2\in \N$ such that,
\begin{equation}\label{eq_proof_ILSE_8}
\forall L\geq 3L_2,\ \mathsf{P}\left( \left|\left| T_{[y+\ell]-\ell}^{-\ell L}(E)\left( \begin{smallmatrix}
0 \\
I_{\mathrm{N}}
\end{smallmatrix} \right) \right|\right| \geq \ee^{(\gamma_1(E)+\varepsilon) 4\ell \frac{L}{3} } \right) \geq 1-\ee^{-2\kappa_0 \ell \frac{L}{3}}.
\end{equation}
If $C=\max(C_1,C_2)$, using (\ref{eq_proof_ILSE_5}), (\ref{eq_proof_ILSE_6}), (\ref{eq_proof_ILSE_7}) and (\ref{eq_proof_ILSE_8}), we get, for $L\geq \max(L_1,3L_3)$,
\begin{equation}\label{eq_proof_ILSE_9}
\mathsf{P}\left(E\notin \sigma(\HL)\ \mathrm{and}\ ||G^{(L)}(E,x,y)|| \leq C\ee^{-(\gamma_1(E)-10\varepsilon) \ell \frac{L}{3}} \right) \geq 1-\ee^{2\kappa \ell L}-\ee^{-2\kappa_0 \ell \frac{L}{3}}.
\end{equation}
Now, if we assume that $x\in [-\ell L,-\ell L +\ell]$ and $y\in [-\ell\frac{L}{3},\ell\frac{L}{3}]$, we have $x\leq y$ and for $E\notin \sigma(\HL)$, $G^{(L)}(E,x,y)=U_-(x)W(U_-,U_+)^{-1}\ {^t}U_+(y)$. In a similar way as we proved (\ref{eq_proof_ILSE_9}), we get the same estimate : 
\begin{equation}\label{eq_proof_ILSE_10}
\mathsf{P}\left(E\notin \sigma(\HL)\ \mathrm{and}\ ||G^{(L)}(E,x,y)|| \leq C\ee^{-(\gamma_1(E)-10\varepsilon) \ell \frac{L}{3}} \right) \geq 1-\ee^{2\kappa \ell L}-\ee^{-2\kappa_0 \ell \frac{L}{3}}.
\end{equation}
\vskip 2mm

\noindent We introduce the events $A_{L,\varepsilon}(E)$ and $B_{L,\varepsilon}(E)$ defined by :
\begin{eqnarray}
A_{L,\varepsilon}(E) & = & \Big\{\omega \in \Omega\ \Big|\ E\notin \sigma(\HL)\ \mathrm{and}\ \nonumber \\
 & & \sup_{x\in \R} \int_{\R} ||\mathbf{1}_{0,L}^{\mathrm{out}}(x) G^{(L)}(E,x,y) \mathbf{1}_{0,L}^{\mathrm{in}}(y)|| \dd y  \leq \ee^{-(\gamma_1(E)-\varepsilon) \ell \frac{L}{3}} \Big\} \nonumber 
\end{eqnarray}
and 
\begin{eqnarray}
B_{L,\varepsilon}(E) & = & \Big\{\omega \in \Omega\ \Big|\ E\notin \sigma(\HL)\ \mathrm{and}\
\nonumber \\
& & \sup_{y\in \R} \int_{\R} ||\mathbf{1}_{0,L}^{\mathrm{out}}(x) G^{(L)}(E,x,y) \mathbf{1}_{0,L}^{\mathrm{in}}(y)|| \dd x  \leq \ee^{-(\gamma_1(E)-\varepsilon) \ell \frac{L}{3}} \Big\}. \nonumber 
\end{eqnarray}
Then, from (\ref{eq_proof_ILSE_9}) and (\ref{eq_proof_ILSE_10}), we deduce that, for every $\varepsilon >0$, there exist $\delta >0$ and $L_0\in \N$, $L_0\geq \max(L_1,3L_2)$, such that,
$$\forall L\geq L_0,\ L\in 3\N,\ \mathsf{P}\left( A_{L,\varepsilon}(E) \cap B_{L,\varepsilon}(E) \right) \geq 1-\ee^{-\delta \ell L}.$$
To pass from this estimate on the kernel $G^{(L)}(E,x,y)$ of $R^{(0,L)}(E)$ to the estimate (\ref{eq_prop_ILSE}) on $R^{(0,L)}(E)$, we use Schur's test. It finishes the proof.
\end{proof}

\noindent It is interesting here to remark that the exponential decaying rate of the resolvent, and thus of the eigenfunctions of $\Ho$, is almost $\ell \gamma_1(E)$, the smallest positive Lyapunov exponent times the interaction length $\ell$. We have now all the requirements needed to prove Theorem \ref{thm_localization} and Theorem \ref{thm_localization_Hl}.

\begin{proof}[Proof of Theorem \ref{thm_localization}]
Let $I\subset \R$ be a compact interval, $\Sigma \cap I \neq \emptyset$, and let $\tilde{I}$ be an open interval, $I\subset \tilde{I}$, such that, for every $E\in \tilde{I}$, $G(E)$ is \LpSI. If we look at the proof of \cite[Theorem A.1]{GK04}, we see that the potential only appears through estimates of its absolute value and so, changing the absolute value into a matrix-norm in this proof, we get that $\Ho$ has the properties (SLI), (EDI), (NE) and (SGEE) on $I$. From the form of the potential of $\Ho$ and the assumption on independence of the $V_{\omega}^{(n)}$, $\Ho$ also has the property (IAD). By Theorem \ref{thm_wegner}, $\Ho$ verifies a Wegner estimate (W) on $I$, and by Proposition \ref{prop_ILSE}, it verifies an (ILSE) estimate on $I$. So we can apply Theorem \ref{thm_MSA_GK} for every $E_0 \in I$ to get $I\subset \Sigma_{\mathrm{MSA}}$. Then, applying Theorem \ref{thm_loc_MSA}, we get that  :
$$I\subset \Sigma_{\mathrm{EL}} \cap \Sigma_{\mathrm{SSEHSKD}} \subset \Sigma_{\mathrm{EL}} \cap \Sigma_{\mathrm{SDL}}.$$
It proves Theorem \ref{thm_localization}.
\end{proof}

\begin{proof}[Proof of Theorem \ref{thm_localization_Hl}]
For the point $(i)$, by Proposition \ref{prop_GE_SpN}, for $\ell <\ell_C$ and for every $E\in I(\ell,N)$, $G(E)$ is \LpSI. As $I(\ell,N)$ is a compact interval, we can apply Theorem \ref{thm_localization} to $\Hl$ on every open interval $I\subset I(\ell,N)$ such that $\Sigma \cap I \neq \emptyset$. We already remark in the introduction that such open intervals exist.

\noindent For the point $(ii)$, we use \cite[Proposition 2]{boumazampag} for the existence of a discrete set $\mathcal{S}\subset \R$ such that for every $E\in (2,+\infty)\setminus \mathcal{S}$, $G(E)$ is \LpSI. Once again, if $I\subset (2,+\infty)\setminus \mathcal{S}$ is a compact interval, as $\mathcal{S}$ is discrete, there always exists an open interval $\tilde{I} \subset (2,+\infty)\setminus \mathcal{S}$, $I\subset \tilde{I}$. Therefore, we can apply Theorem \ref{thm_localization} to $H_{1}(\omega)$ for $N=2$ on $I \subset (2,+\infty)\setminus \mathcal{S}$, $I$ compact and such that $\Sigma\cap I \neq \emptyset$.
\end{proof}



\end{document}